\documentclass[12pt,english]{amsart}
\usepackage{geometry} % see geometry.pdf on how to lay out the page. There's lots.
\usepackage[shortlabels]{enumitem}
\usepackage{calligra}
\usepackage[T1]{fontenc}
\usepackage[latin9]{inputenc}
\usepackage{geometry}
\usepackage{amsthm}
\usepackage{amssymb}
\usepackage{graphicx}
\usepackage{float}
\usepackage{color}
\definecolor{ryan}{cmyk}{0.64,0,0.95,0.40}
\usepackage{courier}
\usepackage{url}
\usepackage{cleveref}
\crefname{section}{§}{§§}
\Crefname{section}{§}{§§}
\crefformat{section}{§#2#1#3}

\usepackage{parskip}
\makeatletter
\numberwithin{equation}{section}
\numberwithin{figure}{section}

  \theoremstyle{plain}
  \newtheorem*{thm*}{\protect\theoremname}
  
\theoremstyle{plain}
\newtheorem{thm}{\protect\theoremname}

  \theoremstyle{definition}
  \newtheorem{defn}[thm]{\protect\definitionname}

  \theoremstyle{plain}
  \newtheorem{lem}[thm]{\protect\lemmaname}

\@ifundefined{showcaptionsetup}{}{%
 \PassOptionsToPackage{caption=false}{subfig}}
\usepackage{subfig}
\makeatother

\usepackage{babel}
  \providecommand{\definitionname}{Definition}
  \providecommand{\lemmaname}{Lemma}
    
  \providecommand{\theoremname}{Theorem}
\providecommand{\theoremname}{Theorem}

\usepackage{amsthm}
\newtheorem{prop}{Proposition}
 
\theoremstyle{remark}
\newtheorem*{remark}{Remark}

\begin{document}

\title{On the spectral properties of Feigenbaum graphs}

\author{Ryan Flanagan, Lucas Lacasa and Vincenzo Nicosia}
\address{School of Mathematical Sciences, Queen Mary University of London, E14NS London (UK)}

\begin{abstract}
A Horizontal Visibility Graph (HVG) is a simple graph extracted from an ordered sequence of real values, and this mapping has been used to provide a combinatorial encryption of time series for the task of performing network based time series analysis. While some properties of the spectrum of these graphs --such as the largest eigenvalue of the adjacency matrix-- have been routinely used as measures to characterise time series complexity, a theoretic understanding of such properties is lacking. In this work we explore some algebraic and spectral properties of these graphs associated to periodic and chaotic time series. We focus on the family of Feigenbaum graphs, which are HVGs constructed in correspondence with the trajectories of one-parameter unimodal maps undergoing a period-doubling route to chaos (Feigenbaum scenario). For the set of values of the map's parameter $\mu$ for which the orbits are periodic with period $2^n$, Feigenbaum graphs are fully characterised by two integers $n,k$ and admit an algebraic structure. We explore the spectral properties of these graphs for finite $n$ and $k$, and among other interesting patterns we find a scaling relation for the maximal eigenvalue and we prove some bounds explaining it. We also provide numerical and rigorous results on a few other properties including the determinant or the number of spanning trees. In a second step, we explore the set of Feigenbaum graphs obtained for the range of values of the map's parameter for which the system displays chaos. We show that in this case, Feigenbaum graphs form an ensemble for each value of $\mu$ and the system is typically weakly self-averaging. Unexpectedly, we find that while the largest eigenvalue can distinguish chaos from an iid process, it is not a good measure to quantify the chaoticity of the process, and that the eigenvalue density does a better job.

\end{abstract}


\maketitle

\section{Introduction}
In recent years, a great deal of attention has been devoted to the construction of graphs associated to time series, with the aims to make network based time series analysis \cite{PR_review}. Here we consider a specific method --horizontal visibility graphs-- by which an ordered sequence of $N$ real-valued data is transformed into a graph with $N$ nodes, whose edges are established among the $N$ nodes according to a given ordering criterion in the sequence \cite{PRE, PNAS}. While a great deal of effort has been paid to study properties of these graphs related to the degree sequence \cite{Luque2016, nonlinearity}, less attention has been paid to their spectral properties. Nevertheless, the so-called Graph Index Complexity (GIC) \cite{GIC}, a rescaled quantity of the maximal eigenvalue of the graph's adjacency matrix, has been proposed as a measure to characterise the complexity of the associated sequence, and has been used in several applications including detection of Alzheimer's disease \cite{Ahmadlou2010} or epilepsy \cite{epilepsy} among others \cite{GIC2, GIC3} or the discrimination between randomness and chaos \cite{Fioriti}. 
However, a basic theoretical understanding of the spectral properties of HVGs is still lacking. This is the main aim of this paper. To achieve this aim, we generate sequences (trajectories) from the logistic map, as this is a well-known map which generates both periodic and chaotic sequences, allowing us to explore spectral properties of HVGs associated to different classes of time series. In previous works, the HVG of a time series generated by the logistic map for a specific value of the parameter $\mu$ was coined as a \emph{Feigenbaum graph} \cite{Feig}. In a nutshell, Feigenbaum graphs are HVGs associated with the Feigenbaum scenario, where one-dimensional unimodal maps exhibit a period-doubling route to chaos. In this work we give a first look at some spectral properties of this family of graphs.

The rest of the paper goes as follows. In section \cref{sec:BC} we define and provide a basic characterisation of Feigenbaum graphs below and above the accumulation point. We show that below the accumulation point, Feigenbaum graphs are easily enumerable in terms of a two-parameter family of graphs which can be generated in terms of two graph operations and admit an algebraic structure. Such enumeration is not possible above the accumulation point, where we show that for particular values of the map's parameter we no longer have unique graphs but an ensemble of them. In sections \cref{sec:regular} and \cref{sec:chaos} we explore the spectral properties including the spectrum of the adjacency matrix --with special interest in the largest eigenvalue-- above and below the accumulation point. For the chaotic region, we finally compare the results to those associated with an iid process. In section \cref{sec:conclusion} we conclude.

%The largest eigenvalue of the adjacency matrix of a graph gives us relevant information about dynamic processes running on top of the graph. For instance in a SIS epidemic model, the epidemic threshold that separates the phases occurs at a spreading rate $\tau \sim \lambda_1^{-1}$. Synchronization properties of the graph are also related to $\lambda_1$ (the transition to synchronization of Kuramoto oscillators again depends on $\lambda_1^{-1}$)

\section{From Horizontal visibility graphs to Feigenbaum Graphs: Basic Characterisation} 
\label{sec:BC}

We start with a few definitions.

\begin{defn}\label{def:HVG}
(Horizontal visibility graph HVG). {\it Let ${\mathcal S}=\{x_{i}\}_{i=1}^N$ be an ordered sequence of $N$ real-valued data $x_i \in \mathbb{R} \  \forall i=1,\dots,N$. Then, the horizontal visibility graph (HVG) \cite{PRE} associated to ${\mathcal S}$ is an undirected graph
of $N$ ordered vertices (where a vertex with ordinal $i$ is related to datum $x_i$), such that
two vertices $i$ and $j$ share an edge iff $x_{i},x_{j}>x_{m}$ for all $m$ such
that $i<m<j$.}
\end{defn}

%{\bf Definition} (Horizontal visibility graph HVG)
%Let ${\mathcal S}=\{x_{i}\}_{i=1}^N$ be an ordered series of $N$ real-valued data $x_i \in \mathbb{R} \  \forall i=1,\dots,N$. Then, the horizontal visibility graph (HVG) \cite{luque2009horizontal} associated to ${\mathcal S}$ is an undirected graph
%of $N$ ordered vertices (where a vertex with ordinal $i$ is related to data $x_i$), such that
%two vertices $i$ and $j$ share an edge iff $x_{i},x_{j}>x_{k}$ for all $k$ such
%that $i<k<j$.

It was proved that HVGs are always outerplanar \cite{severini}, and it is easy to see that, since the order in the sequence (time series) yields a natural label of the vertices, by construction HVGs contain a `trivial' Hamiltonian path given by the sequence $(1,2,\dots,N)$. Furthermore, the mean degree $\bar d$ of a HVG associated to a periodic series with period $T$ is \cite{Feig}
\begin{equation}
{\bar d}=4\bigg(1-\frac{1}{2T}\bigg)
\label{periodic_k}
\end{equation}
Here we consider the set of HVGs generated from trajectories of the well-known logistic map $x_{t+1}=\mu x_{t}(1-x_{t})$ where $\mu\in[0,4]$ is a parameter and $x_{t}\in[0,1]$. This is a unimodal map that undergoes a period-double bifurcation route to chaos as the parameter $\mu$ is increased. For $\mu < \mu_{\infty}\approx 3.569...$ the attractive set consists of periodic orbits with period $2^n$, where $n$ is an integer that increases without bounds as $\mu$ approaches the accumulation point $\mu_{\infty}$. For any given integer $n\geq 0$, one can associate a range of values $I_{n}=[\mu_n,\mu_{n+1})$ where $\mu_n$ is the value of the map's parameter for which a stable periodic orbit of period $2^n$ first appears (with $\lim_{n \to \infty} \mu_n \equiv \mu_{\infty}$). By construction, we have a bijection between $I_n$ and $\mathbb{N}$.\\
HVGs generated from trajectories of the logistic map have been studied before, and have been coined as {\it Feigenbaum graphs} \cite{Feig}. We start by formally introducing these:

\begin{defn}(The Infinite Feigenbaum graph)\label{def:infFBG}
{\it Consider a periodic orbit of period $T=2^n$ from the logistic map, and build a time series of $2NT+1$ data (with $N\in\mathbb{N}$), as 
\[
\mathcal{S}=\{ x_{-NT},x_{-N{T}+1},\dots,x_{-1},x_0,x_1,\dots,x_{N{T}-1},x_{N{T}}\}.
\]
\noindent The associated HVG is referred to as a Feigenbaum graph \cite{Feig}. In the limit $N\to\infty$, the associated HVG is a locally finite infinite graph. It is denoted $F_n^\infty$ and is referred to as the infinite Feigenbaum graph \cite{Feig}.}
\end{defn}

\begin{remark}
As the infinite Feigenbaum graph is a connected, locally finite, infinite graph, it is countable \cite{graph_wilson}.
\end{remark}
A sketch of $F_n^\infty$ for a few values of $n$ is depicted in figure \ref{fig:Periodic-Feigenbaum-graphs}.

\begin{figure}[h]
\begin{centering}
\includegraphics[scale=0.7]{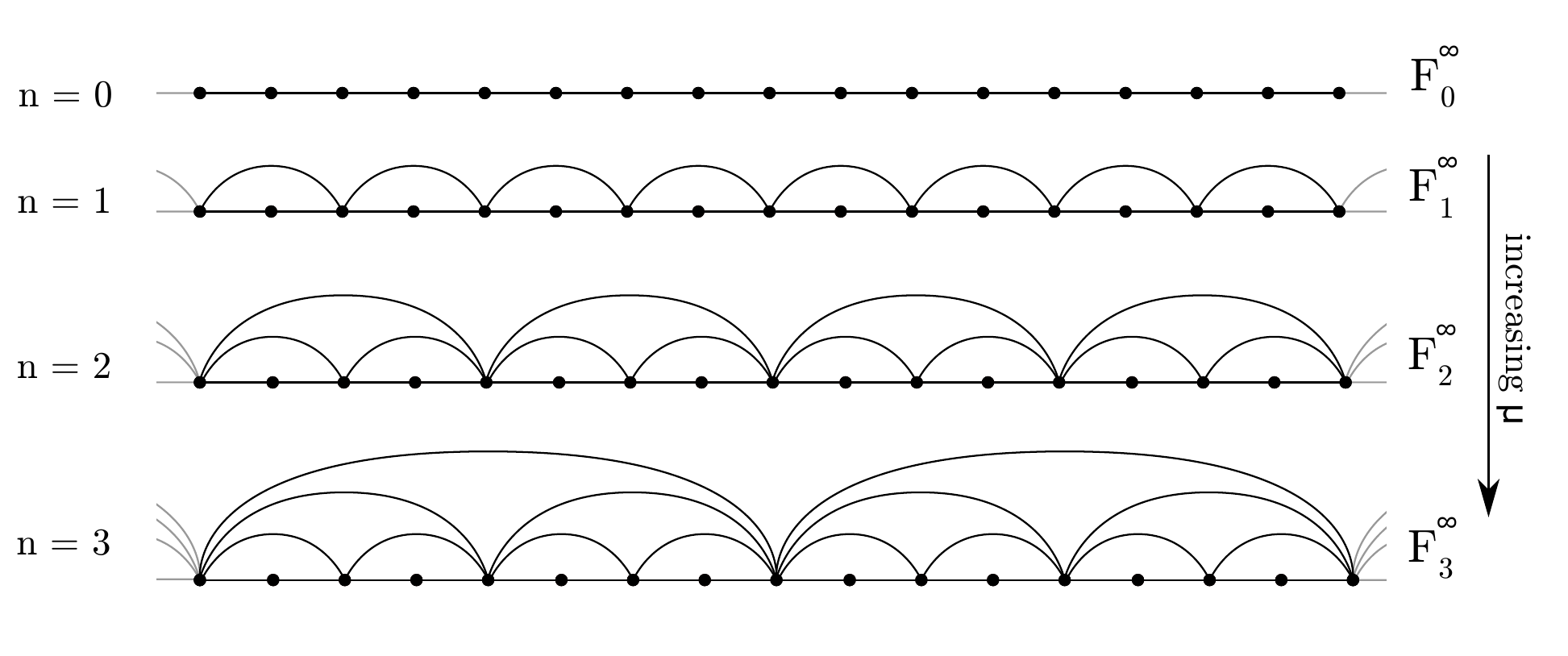}
\protect\caption{\label{fig:Periodic-Feigenbaum-graphs}Sketch of the family of (infinite) Feigenbaum graphs $F_n^\infty$ for $\mu<\mu_{\infty.}$, displaying the sequence of graphs associated to the periodic attractors of increasing period $T=2^{n}$ of an unimodal map undergoing a period doubling cascade.}
\par\end{centering}
\end{figure}

\subsection{Feigenbaum graphs with $\mu<\mu_{\infty}$: a simple parametrisation $F_n^k$}\label{sec:finFBG}
Observe that for any $n<\infty$ (that is, for $\mu<\mu_{\infty}$), the trajectory generated by the logistic map is --after an irrelevant transient-- a periodic series. In these cases, the Feigenbaum graph is built as a concatenation of identical subgraphs (see figure \ref{fig:Periodic-Feigenbaum-graphs}).
We label the {\it motifs} which build these graphs as $F_{n}$, and for illustration purposes we show in figure \ref{fig:Single-motifs-} the first four of them.\\
For a fixed $n$, we can then `concatenate' motifs (in a way which will be formally defined later) and the graph resulting of concatenating $k$ motifs is denoted by $F_{n}^{k}$ (so that $F_{n}^{1}=F_{n}$ and $\lim_{k\to \infty}F_n^k=F_n^\infty$). Whereas in \cite{Feig} a Feigenbaum graph was defined for a bi-infinite trajectory ($k\to \infty$), one can however extend this definition to {\it finite} graphs by fixing a finite $k$. Accordingly, the elements in the bi-parametric set $\{F_n^k\}_{n\geq 0, k \geq 0}$ (where $F_n^1 := F_n$)  provides a useful enumeration of finite Feigenbaum graphs. For completeness, we define $F_n^0$ to be the empty graph of one node. With a little abuse of language, in what follows we will indistinctively refer to $F_n^k$ and $F_n^\infty$ as Feigenbaum graphs.\\

\begin{figure}[h]
\begin{centering}
\includegraphics[scale=0.3 ]{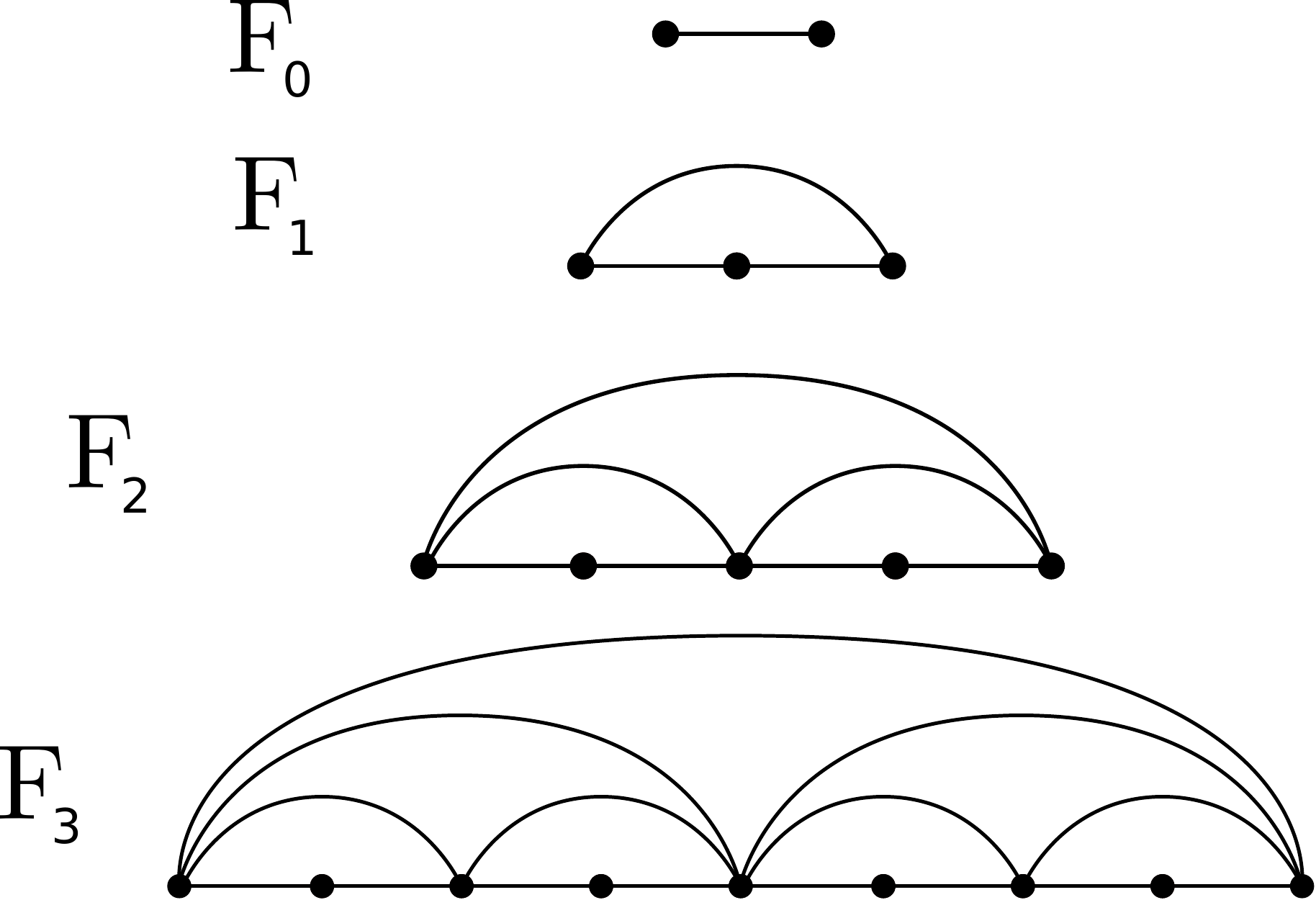}
\par\end{centering}
\protect\caption{\label{fig:Single-motifs-}Single motifs $F_{n}^1\equiv F_n$ of the HVGs associated to the logistic maps with period $T=2^{n}$.}
\end{figure}

\begin{remark}
Given an integer $n$, both $F_n^k$ and $F_n^\infty$ are unique $\forall\mu \in I_n$: for the range of values of $\mu$ for which the map is periodic and the associated time series has the same period, the resulting Feigenbaum graph is unique, i.e. it is  not dependent on the map's initial condition. This observation, as we shall see, does not hold for the range of values of $\mu$ that correspond to chaotic behaviour.
Furthermore, the hierarchy of Feigenbaum graphs is universal for all unimodal maps undergoing a Feigenbaum scenario. In particular, this means that this hierarchy is not only associated to the logistic map but to any unimodal map. The reason is because Feigenbaum graphs are based in the order of visits to the stable branches and this order is unique for all unimodal maps. 
\end{remark}

\noindent Note that we can generate all the elements of the family $\{F_n^k\}_{n,k \geq 0}$ by combining them using two graph-theoretical operations which we now define:

\begin{figure}
\begin{centering}
\includegraphics[scale=1]{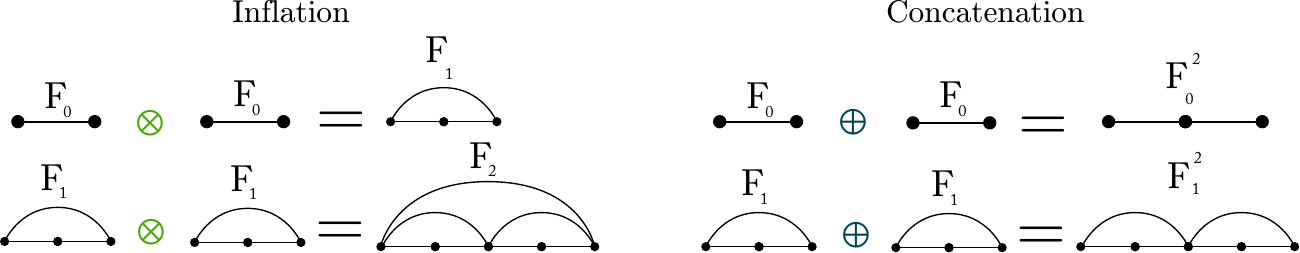}
\par\end{centering}
\protect\caption{\label{fig:Rules}A visualisation of the motif inflation ($\otimes$) and motif concatenation ($\oplus$) rules.}
\end{figure}

\begin{defn}(Motif inflation $\otimes$)
Consider two undirected graphs $G_1=(V_1,E_1)$ and $G_2=(V_2,E_2)$, where $V_i$ are the vertex sets ($|V_i|=N_i$) and $E_i$ are the edge sets, where $V_i$ are totally ordered. We label the vertex set of $G_1$ by $V_1=(1,2,\dots,N_1)$ and similarly for $G_2$ we have $V_2=(1',2',\dots,N'_2)$. Then $G_1 \otimes G_2$ is a graph which fulfils the following conditions:

\begin{enumerate}
  \item $G_1 \otimes G_2$ is a graph with $N_1+N_2-1$ vertices,
  \item whose vertex set $V''=(1,2,\dots,N_1-1,1'',2',3',\dots,N_2')$,
  \item where vertex $1''$ is a block vertex that merges the vertices $N_1$ and $1'$ (from $V_1$ and $V_2$ respectively), and inherits all the edges that were incident to  both of them.
  \item The vertices $1$ and $N'_2$ share an edge in $G\otimes G$.
  \item The remaining edge set is formed by all edges between vertices $\{2,3,\dots,N_1-1\}$ inherited from $G_1$ and between the vertices $\{2',3',\dots,[N_2-1]'\}$ inherited from $G_2$.
  \end{enumerate}
\end{defn}
For illustration, a visualisation of the inflation operation is shown in the left panel of figure \ref{fig:Rules}. 

By induction, one can then easily prove that $F_{n+1}=F_{n}\otimes F_{n}$. Let us define ${\bf A}_n$ as the adjacency matrix of $F_n$ (defining the adjacency matrix ${\bf A}=\{a_{ij}\}$ to be a binary matrix which assigns $a_{ij}=1$ if $i$ and $j$ are two nodes linked by an edge, and zero otherwise). The adjacency matrix ${\bf A}_{n+1}$ of $F_{n+1}$ can be expressed in terms of the adjacency matrix ${\bf A}_n$ of $F_n$ as illustrated in figure  \ref{fig:An}. Therefore, starting from $k=1$, the operation $\otimes$ iteratively generates all the elements of the set $\{F_n^{1}, n>0\}$. This means that $(F_n,\otimes)$ is a unary system if we interpret $\otimes$ as a unary operation $\otimes: \{F_n\} \to \{F_n\}$.
Notice however that this set is not closed under $\otimes$, as for $n_1 \neq n_2, \ \nexists \  n_3 >0 $ such that $F_{n_1} \otimes F_{n_2} = F_{n_3}$. The graphs formed by combining together $F_{n_1}$ and $F_{n_2}$ with $n_1\neq n_2$ are indeed not Feigenbaum graphs, but are still HVGs, hence the set of all HVGs is closed under this operation.

\begin{figure}
\begin{centering}
\includegraphics[scale=0.7]{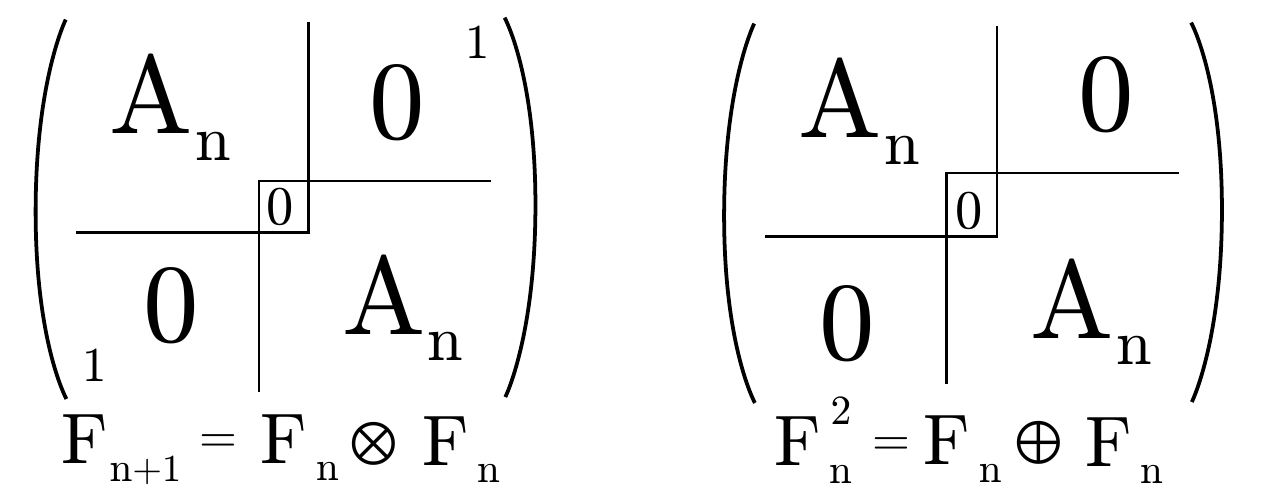}
\par\end{centering}
\protect\caption{\label{fig:An} The adjacency matrices of  $F_{n+1}= F_{n} \otimes F_{n}$ and $F_n^{2}=F_n \oplus F_n$ expressed in terms of ${\bf A}_n$. These composition rules are labelled as a graph inflation and concatenation respectively.}
\end{figure} 

\begin{defn}\label{concatenation}  (Motif concatenation $\oplus$) Consider two undirected graphs $G_1=(V_1,E_1)$ and $G_2=(V_2,E_2)$, where $V_i$ are the vertex sets ($|V_i|=N_i$) and $E_i$ are the edge sets, where $V_i$ are totally ordered. We label the vertex set of $G_1$ by $V_1=(1,2,\dots,N_1)$ and similarly for $G_2$ we have $V_2=(1',2',\dots,N'_2)$. Then $G_1 \oplus G_2$ is a graph which fulfils the following conditions:

\begin{enumerate}
  \item $G_1 \oplus G_2$ is a graph with $N_1+N_2-1$ vertices,
  \item whose vertex set $V''=(1,2,\dots,N_1-1,1'',2',3',\dots,n_2')$,
  \item where vertex $1''$ is a block vertex that merges the $N_1$ and $1'$ (from $V_1$ and $V_2$ respectively), and inherits all the edges that were incident to  both of them,
  \item The vertices $1$ and $N_2$ \textbf{do not} share an edge in $G_1 \oplus G_2$,
  \item The remaining edge set is formed by all edges between vertices $\{2,3,\dots,N_1-1\}$ inherited from $G_1$ and between the vertices $\{2',3',\dots,[N_2-1]'\}$ inherited from $G_2$.
  \end{enumerate}
  \end{defn}
  
A visualisation of this rule, both graphically and algebraically, is shown in the right panels of Figure \ref{fig:Rules} and \ref{fig:An} respectively. We also define ${\bf A}_n^k$ to be the adjacency matrix of $F_n^k$. To avoid confusion, we also state that, using parentheses, $({\bf A}_n^k)^p$ is the $p$th power of the corresponding adjacency matrix\\
One can easily see that, locally, the inflation rule on two graphs $G_1 \otimes G_2$ is equivalent to the concatenation one $G_1 \oplus G_2$ if we add an extra edge between the first and last vertex. Now, for any given $n$, one has $F_n^k \oplus F_n^1= F_n^{k+1}$. It is also easy to prove that $\forall n\geq0, \  k_1,k_2\geq0, F_n^{k_1} \oplus F_n^{k_2} = F_n^{k_1+k_2}$. Therefore, for a fixed $n>0$, the operation $\oplus$ generates all the elements of the set $\{F_n^k, k>0\}$. It is also easy to prove that, for a fixed $n>0$, $(F_n^k,\oplus)$ is a commutative monoid with the identify element being the empty graph of one node $F_n^0$, hence  $(F_n^k,\oplus)$ is isomorphic to $(\mathbb{N},+)$ .\\ 
%\textit{Proof}:
%\begin{itemize}
%\item identity element: $F_n^0$ is the empty graph of one node {\color{red}{(or alternatively this can be defined as the null graph (empty graph of zero nodes))}}. {\color{blue}{with the way we have defined $\oplus$ I don't think it can be defined as the null graph, or maybe it can, but it might be safer to use the graph of one node}}
%\item associativity: $\forall k_1,k_2,k_3>0 \ (F_n^{k_1} \oplus F_n^{k_2}) \oplus F_n^{k_3}=F_n^{k_1+k_2+k_3}=F_n^{k_1} \oplus (F_n^{k_2} \oplus F_n^{k_3})$, where $+$ is standard addition between naturals.
%\item commutativity: $F_n^{k_1} \oplus F_n^{k_2}= F_n^{k_1+k_2}= F_n^{k_2} \oplus F_n^{k_1}$
%\end{itemize}

%{\bf Proposition }  For a fixed $n>0$, $(F_n^k,\oplus)$ is isomorphic to $(\mathbb{N},+)$ . 
%{\color{red}{This is intuitive but we need to check}}

\begin{remark}
Note that the set $\{F_n^k\}|_{n,k}$ can also be created with the aid of simplicial complexes. Given an arbitrary $F_n$, we can create $F_{n+1}$ by gluing a triangle (i.e., a $2$-simplex) to the edges attached to each node with degree $2$. We show $F_3$ in the standard way, with the corresponding simplicial complex representation and equivalent nodes, in Fig.~\ref{SCFB3}. In the same figure we also depict $F_4$, with the $8$ $2$-simplices glued to the edges attached to each node with degree $2$ in $F_3$.
\end{remark}

\begin{figure}
\begin{centering}
\includegraphics[scale=0.7]{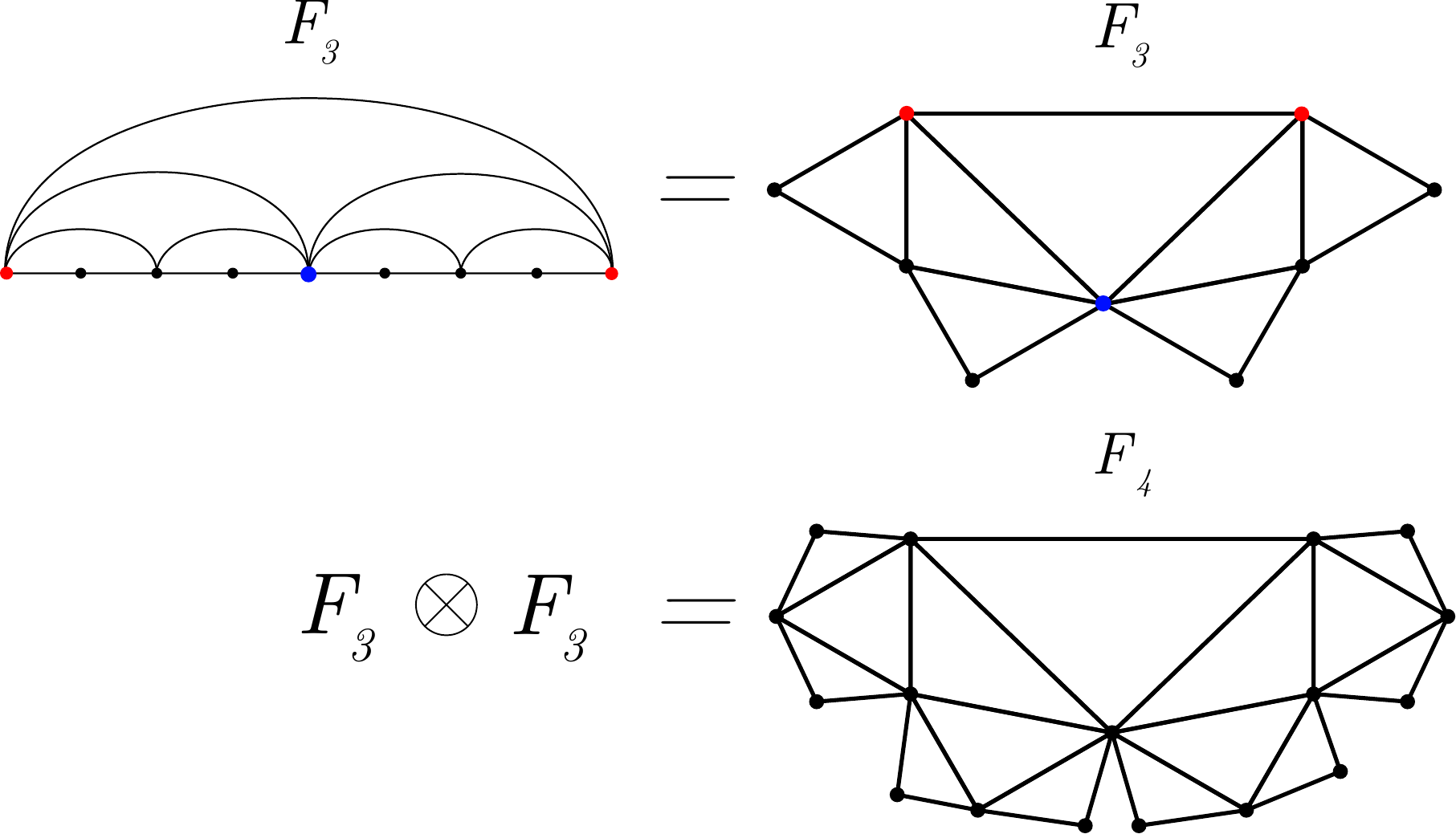}
\par\end{centering}
\protect\caption{\label{SCFB3}A visualisation of the standard layout of $F_3$ (top left) with its simplicial complex layout (top right), along with equivalent nodes in red and blue. On the bottom we plot $F_4$, and we can see the $8$ $2$-simplices glued to the edges attached to each node with degree $2$.}
\end{figure}

The set $\{F_n^k\}|_n$ (fixed $n$) is finitely generated by $F_n^1$ under $\oplus$, while the generating set of $\{F_n\}$ under $\otimes$ is $\{F_1\}$. 
Now, consider the larger set $\{F_n^k,n,k \in \mathbb{N}\}$ where $n$ and $k$ are now free parameters. This set contains the two-parameter $(n,k)$ family of Feigenbaum graphs. This set is again finitely generated by $F_1$ using the operations $\oplus$ and $\otimes$.  
Exploration of the algebraic properties of $\{F_n^k,n,k \in \mathbb{N}\}$ is an interesting topic for future research. However, here we are interested in the spectral properties of $F_n^k$. Some very basic observations, which will be helpful later in the task of bounding eigenvalues, are summarised in the following proposition.

\begin{prop}
\label{prop:initial}
Consider the set of graphs $F_n^k$, for $n,k \in \mathbb{N}$, and let $V_n^k$ and $E_n^k$ be the size of the vertex and edge set respectively, with $V_n:=V_n^1$, $E_n=E_n^1$. Then the following holds:
\begin{enumerate}
\item ${F}_n$ is a graph with $V_n=2^n+1$ vertices and $E_n=2^{n+1}-1$ edges.
\item ${F}_n^k$ is a graph with $V_n^k=2^nk+1$ vertices and $E_n^k=k(2^{n+1}-1)$ edges. %\ryan{added number of edges here}
%\item The largest degree is \textcolor{red}{FINISH THIS}
%\item The degree sequence is  \textcolor{red}{FINISH THIS}
%\item Mean path length  \textcolor{red}{see paper chaos}
%\item The diameter $D$ is 
%\item The degree distribution (see paper Chaos) is
\end{enumerate}
\end{prop}
\begin{proof} The proof trivially follows from the definitions of $\otimes$ and $\oplus$.
\end{proof}

The spectral properties of ${F}_n^k$ will be addressed in \cref{sec:regular}. For readability, we will split this initial study in two natural directions: in \cref{sec:n} we set $k=1$ and consider the spectral properties of ${F}_n$ (i.e., for $k=1$, as $n$ increases), whereas in \cref{sec:k} we set $n$ fixed and consider the spectral properties of ${F}_n^k$ as $k$ increases, i.e. the finite size truncations of infinite Feigenbaum graphs. Finally in \cref{sec:altogether} we will explore the spectrum of $F_n^k$ when $n$ and $k$ are finite and both vary.

\subsection{The large $n$ and $k$ limits}
The variables $n$ and $k$ have clear, different meanings: $n$ is related to the period $T$ of the logistic map's trajectories via $T=2^n$ (physically speaking, $n$ is related to $\mu$ ). In particular, the period-doubling bifurcation cascade that the logistic map experiences relates to successive increases of $n$, where the onset of chaos ($\mu=\mu_{\infty}$) is only reached in the limit $n \to \infty$. On the other hand,  $k$ is a parameter that only describes the length of the trajectory (and therefore properties of a trajectory, for example its periodicity, will only be revealed when $k$ is large, or in the limit $k \to \infty$),  in particular $k$ is the number of concatenated motifs and is related to the size $N$ of the trajectory (the length of the time series) via $N=V_n^k=2^nk+1$.  Note that a priori we have two possible ways to take the limits of large $n$ and $k$. On the one hand, we can fix $n$ and let $k \to \infty$. This mimics the situation where we have an infinitely long trajectory of finite period $T=2^n$. In this limit, $F_n^k$ is by construction a locally finite infinite graph, i.e. the number of vertices is infinite but each vertex has a finite number of edges.

On the other hand, we can also fix $k$ (e.g. $k=1$) and take $n\to \infty$. This mimics the situation where only a single `period' is extracted from the series, however as this period is $T=2^n$, in the limit the time series is infinitely long, obtaining an infinite graph. However, in this limit the graph is {\it not} locally finite: as we will show later in Proposition~\ref{prop:basic} the degree of the central vertex of $F_n$ increases linearly with $n$, so there are at least $k$ vertices in $F_n^k$ whose degree increases (without bound) with $n$. On the other hand, this is still a countable infinite graph. 

Therefore, taking the limits $k\to \infty$ and $n\to \infty$ yield different types of infinite graphs: a locally finite infinite graph in one hand and a countable infinite graph on the other. In particular, the fact that the limit $n\to \infty$ yields infinite graphs which are not locally finite has important consequences for the spectral properties of these graphs. Recall that for finite graphs, the spectrum of a graph is simply the set of all eigenvalues of the respective adjacency matrix $\bf A$. However if the graph is infinite, the spectrum of $\bf A$ depends on the choice of the space on which $\bf A$ acts as a linear operator (typically one considers the Hilbert space ${\ell}^2(V)$, where $V$ is the set of vertices). It is well known that if the infinite graph is locally finite, then $\bf A$ acts on  ${\ell}^2(V)$ as a self-adjoint operator and its norm is smaller or equal to $d_{\text{max}}$, the largest degree of the graph \cite{spectral_infinite}. If the property of local finiteness is relaxed, then this operator is not bounded anymore. 
%\textcolor{red}{and is neither compact} \ryan{ryan: the operator is is compact iff G has finitely many edges (thm 3.2 in \cite{spectral_infinite}), which is not true in the case of either $F_n^{\infty}$ or $F_{\infty}^k$, so we can remove the part in red}. 
Incidentally, one could create a self-adjoint compact operator on ${\ell}^2(V)$ from an adjacency matrix ${\bf A}=(A_{ij})_{i,j\in\mathbb{N}}$, even if the respective graph is not locally finite, by using the approach of Torgasev \cite{torgasev}: let $c\in(0,1)$ and label the vertices of the graph $V=\{v_1,v_2,\dots\}$ (note that one can always do this as this set is countable). Define define the matrix  ${\bf B}_c=(b_{ij})_{i,j\in\mathbb{N}}$ with $$b_{ij}=A_{ij}\cdot c^{i+j-2}.$$
The matrix ${\bf B}_c$ is a self-adjoint and compact operator on ${\ell}^2(V)$, which is Hilbert-Schmidt and therefore enables the use of the well-developed field of spectral theory. The drawback is that the spectrum arbitrarily depends on both the labelling of the graph and on the constant $c$.

%\textcolor{red}{For what said above, it is not obvious to take the limit $\lim_{n,k\to \infty} F_n^k$, which strictly speaking is the one we should take to explore the onset of chaos $\mu=\mu_{\infty}$. We leave these as interesting open problems, and from now on we will assume that both $n$ and $k$ are finite.}
%\ryan{ryan: unsure about the paragraph in red, changed to the paragraph below}

In summary, the limit $\lim_{n,k\to \infty} F_n^k$ (which is the one we should take to explore the onset of chaos $\mu=\mu_{\infty}$) is non-trivial. For this reason, we leave these as interesting open problems, and from now on we will assume that both $n$ and $k$ are arbitrary large but finite.

\subsection{Feigenbaum graphs with $\mu>\mu_{\infty}$: Chaotic Feigenbaum graph ensembles}
\label{sec:intro_chaos}
In the range $\mu>\mu_{\infty}$, the trajectories of the logistic map are typically chaotic (except for the so-called windows of periodicity, which are essentially subintervals where the period-doubling cascade is self-similarly reproduced albeit with an initial period larger than one). The first observation is that in the chaotic regime the graphs can no longer easily be enumerated. In fact, for a given $\mu$ in the chaotic range the Feigenbaum graph is no longer unique: each different condition will typically generate a different chaotic trajectory and therefore a different Feigenbaum graph. Hence each value of $\mu$ spans a different {\it ensemble} of Feigenbaum graphs, generated by sampling different initial conditions in the map. As discussed in \cref{sec:finFBG}, this is at odds with the case $\mu<\mu_{\infty}$, where for any particular $\mu$ all realisations in an ensemble associated to $\mu$ yielded the same Feigenbaum graph, therefore the ensemble was fully degenerate in that case.
%\ryan{ryan: the above sentence makes it sound like all realisations with $\mu<\mu_{\infty}$ yield the same graph, which is not true, so I changed the sentence.} \\

Of course as the length of the time series approaches infinity, the statistical properties of two different chaotic trajectories extracted at the same value of $\mu$ 
%\textcolor{red}{(two realisations)}\ryan{ryan:remove part in red, didn't say just say they they are two different trajectories?} 
are asymptotically identical, so we expect some kind of statistical equivalence in the resulting Feigenbaum graphs. For instance, for $\mu=4$ (fully developed chaos) one can compute the degree distribution of the (ensemble of) Feigenbaum graphs. This is a statistical quantity which can be solved analytically by using a diagrammatic technique \cite{nonlinearity}, and has been shown to be a valid limit for single realisations. However, in this work we are interested in studying the spectral properties of Feigenbaum graphs, so we need to address whether these properties are sufficiently `robust', i.e. we should check whether these properties do not change much between realisations. This naturally leads to the concept of {\it self-averaging quantities}, which will be investigated in section \cref{sec:chaos1}. Then, in sections \cref{sec:chaos2} we will try to relate the properties of the time series to spectral properties of the graphs. 
%At odds with other properties extracted from the degree sequence \cite{wolfram}, we advance that no obvious spectral property quantifies the `degree of chaoticity' of the time series, and therefore using these to e.g. describe chaotic processes is not a promising avenue.%\ryan{ryan: should this be here? sounds like something that should be in the abstract}

\section{The case $\mu<\mu_{\infty}$: Spectral properties in the period-doubling cascade}
\label{sec:regular}
Here we explore the spectral properties of $\{F_n^k\}$. In particular, we will focus on the maximal eigenvalue of the adjacency matrix of $F_n^k$, although other properties will also be considered, such as the full spectrum, the determinant, and the tree number. For convenience, we split this section in three main blocks: the first explores the dependence of $n$ by focusing on the properties of $\{F_n\}_{n\geq0}$. The second  focuses on the dependence on $k$ by exploring properties of $\{F_n^k\}_{k\geq0}$ where $n$ is fixed. Finally we explore $\{F_n^k\}_{n\geq0,k\geq0}$, where both $n$ and $k$ can vary. 
%\ryan{ryan: added some set notation, in the subscript, as $\{F_n\}$ isn't a set but $\{F_n\}_{n\geq0}$ is}

\subsection{A first view on the full spectrum of $F_n$}
Here we fix $k=1$ and consider the set $\{F_n\}_{n\geq0}$, and we start by exploring the full spectrum the adjacency matrices $\{{\bf A}_n\}_{n\geq0}$ (i.e. the set of $2^n+1$ eigenvalues) associated to $F_n$. The first quantity worth exploring is the number of distinct eigenvalues of ${\bf A}_n$, labelled $q({\bf A}_n)$. To bound this, it is useful to resort to the {\it diameter} $D_n$ of $F_{n}$, defined as
$D_{n}=\max_{i,j}\{\delta_{ij}\}$, where $\delta_{ij}$ is the (shortest path) distance between node $i$ and node $j$. A well known result is $q({\bf A}_n) \geq D_n+1$.
The following theorem provides the diameter $D_n$:

\begin{thm}(Diameter of $F_n$) 
\label{theorem:diameter} The diameter of $F_{n}$ is
$
  D_n = n.
$
\end{thm}
\begin{proof} The proof requires a Lemma. Let us consider $F_{n}$, whose nodes are labelled $\{0,1,2,\dots,2^n\}$, and denote by $R_{n}=2^n$ the rightmost node of $F_{n}$. In the following, we call \textit{extremal points} of $F_n$ the nodes $0$, $R_{n-1}$, and $R_{n}$. Notice that $R_{n-1}$ is the middle-point of the Hamiltonian path of $F_{n}^1$ that starts at node $0$ and proceeds by increasing node labels. Notice as well that since $F_n=F_{n-1}\otimes F_{n-1}$, $F_{n-1}$ has again three extremal points, which are labeled (with respect to $F_n$) either $\{0,R_{n-2}, R_{n-1}\}$ or $\{R_{n-1},\tilde{R}_{n-2},R_n\}$.\\
We denote by $\delta_{ij}$ the distance between node $i$ and node $j$, i.e., the length of the shortest path from $i$ to $j$. We can prove the following Lemma

\begin{lem} (Distance to the closest extremal point.)
\label{eq:lemma}
Consider the graph $F_{n}$ and the minimal distance between a generic node $i$ and the closest extremal point $E$
$$\delta^n_{i,E}=\min\{\delta_{i,0}; \delta_{i,R_{n-1}}; \delta_{i,R_{n}}\}.$$ Then for $n\ge
1$ we have:
\begin{equation}
  \delta_{i,E}\le \left\lfloor\frac{n}{2}\right\rfloor, \forall i\in[0, R_n].\nonumber
  %\label{eq:lemma}
\end{equation}
\end{lem}
%\lucas{lucas: numerically I think the correct bound is $\delta_{i,E}\le \left\lfloor\frac{n}{2}\right\rfloor$. That would make Enzo's proof basically conducing to $\delta_{ij}\le n$. But SINCE we can always find (constructively) $l,m$  such that the bound is tight i.e. $\delta_{l,m}= n$, and because $D_n=\max \{\delta_{i,j}\}$, then we could conclude that $D_n=n$. }

 \begin{proof}We will prove this by strong induction on $n$. When $n=1$, $F_{n}$ is a triangle whose diameter is $D_1 = 1$ and thus all three nodes are extrema, i.e. $\delta_{i,E}=0=\lfloor\frac{1}{2}\rfloor$.
  %\lucas{lucas: for $F_1$, each of the three nodes are defined as extrema by Enzo's definition so $\delta_{i,E}=0$ in that case! Which by the way complies with my formula above (floor function instead of ceiling)}. 
  Now let us assume that the Lemma is valid up to $n-1$ and let's prove it for $n$. Without loss of generality, we assume $i\in[0, R_{n-1}]$. In this case we have that $\delta^n_{i,E} = \min\{\delta_{i,0}, \delta_{i, R_{n-1}}\}$, since $R_{n}$ will be at least one hop farther away from $i$ than either $0$ or $R_{n-1}$. Let us consider first the case where $i\in[0, R_{n-2}]$. In this case $i$ is closer to $0$ than to $R_{n-1}$ (or at most, at the same distance from either of the two), hence $\delta^n_{i,E}=\min\{\delta_{i,0},\delta_{i,R_{n-1}}\}=\delta_{i,0} \le 1 + \delta^{n-2}_{i,E} \le 1 + \left\lfloor\frac{n-2}{2}\right\rfloor$. The first inequality is due to the fact that node $0$ is an extremal point, and the distance from $i$ to $0$ will be either equal to $\delta^{n-2}_{i,E}$ or to $1 + \delta^{n-2}_{i,E}$. The second inequality is due to the induction assumption that Eq.~(\ref{eq:lemma}) is valid up to $n-1$. If $n$ is even, we have: $\delta^{n}_{i,E}\le 1+\frac{n-2}{2} = \frac{n}{2} = \left\lfloor \frac{n}{2}\right\rfloor$. If $n$ is odd instead, we have: $\delta^{n}_{i,E}\le 1+\frac{n-3}{2} = \frac{n-1}{2} = \left\lfloor\frac{n}{2}\right\rfloor$. The case where $i\in[R_{n-2}, R_{n-1}]$ is similar, since we can relabel each node in $[R_{n-2}, R_{n-1}]$ according to the function $\phi(i) = i - R_{n-2}$, and repeat the same reasoning. In conclusion, $\delta^{n}_{i,E}\le\left\lfloor \frac{n}{2}\right\rfloor$ for all $i\in[0, R_{n-1}]$. But since the graph is symmetric around $R_{n-1}$, we have $\delta^{n}_{i,E}\le\left\lfloor \frac{n}{2}\right\rfloor$ for all $i$ in $[0, R_n]$. \end{proof}
 
%\ryan{ryan: this is strong induction no? assuming it's true for $0<k<n-1$ and we want to prove it for $k=n$, so we don't actually need to prove it for $n=1$, so we can probably remove the second sentence}

We can now finish the proof of Theorem \ref{theorem:diameter}. Let us consider two generic nodes $i$ and $j$ in $F_{n}$. First consider the case where $i\in[0, R_{n-1}]$ and $j\in[R_{n-1}, R_{n}]$. We have two possibilities for $\delta_{i,E}$ (either $\delta_{i,E}=\delta_{i,0}$ or $\delta_{i,E} = \delta_{i,R_{n-1}}$) and two possibilities for $\delta_{j,E}$ (either $\delta_{j,E}=\delta_{j,R_{n}}$ or $\delta_{j,E} = \delta_{j,R_{n-1}}$). So we have that
\begin{equation} \nonumber
  \delta_{ij} = \min\left\{
  \begin{array}{l}
    \delta_{i,0} + 1 + \delta_{j,R_{n}},\\
    \delta_{i,0} + 1 + \delta_{j, R_{n-1}},\\
    \delta_{i,R_{n-1}} + 1 + \delta_{j, R_{n}},\\
    \delta_{i,R_{n-1}} + \delta_{j, R_{n-1}}.\\
  \end{array}\right.
\end{equation}
This yields
\begin{equation} \nonumber
  \begin{array}{rl}
    \delta_{ij} &= \min\{\delta_{i,E} + 1 + \delta_{j,E}, \delta_{i,E} + \delta_{j,E}\} \\
    & = \delta_{i,E} + \delta_{j,E} \\
    & \le 2 \left\lfloor\frac{n}{2}\right\rfloor \le n\\
   % &\textcolor{red}{\le} n
  \end{array}
\end{equation}
where we have used Lemma (\ref{eq:lemma}).  Conversely, if we have that $i,j\in
[0, R_{n-1}]$ (or equivalently, both $i,j\in[R_{n-1},R_{n}]$) then:
\begin{equation} \nonumber
  \delta_{ij} = \min\left\{
  \begin{array}{l}
    \delta_{i,R_{n-1}} + 1 + \delta_{j,0},\\
    \delta_{i,0} + 1 + \delta_{j, R_{n-1}},\\
    \delta_{i,0} + \delta_{j,0},\\
    \delta_{i,R_{n-1}} + \delta_{j, R_{n-1}}.\\
  \end{array}\right.
\end{equation}
With a similar argument as above, we get
\begin{equation} \nonumber
  \begin{array}{rl}
    \delta_{ij}&=\min\{\delta_{i,E}+d_{j,E}, 1 + \delta_{i,E} + \delta_{j,E}\}\\
    &= \delta_{i,E}+\delta_{j,E}\\
     & \le 2 \left\lfloor\frac{n}{2}\right\rfloor \le n,\\
  \end{array}
\end{equation}
%

%has $2^{n}+1$ distinct eigenvalues, i.e., all eigenvalues are distinct (as the longest walk can be taken from the left most vertex to the right most vertex along the bottom edges for which there are $2^n$ many). 
%\textcolor{blue}{WHAT IS THIS FOR? WILL BE REMOVED In addition, from \cite{biggs1993algebraic} we have that the sum of eigenvalues is zero, the sum of the squares is twice the number of edges, and the sum of the cubes is six times the number of triangles, hence we have that for $F_n$:
%\begin{itemize}
%\item  $\displaystyle\sum_{i=1}^{2^{n}+1}\lambda_i=0$
%\item  $\displaystyle\sum_{i=1}^{2^{n}+1}{\lambda_i}^2=2(2^{n+1}-1)$
%\item  $\displaystyle\sum_{i=1}^{2^{n}+1}{\lambda_i}^3=6(2^{n}-1)$
%\end{itemize}}

i.e., $\delta_{ij}\le n$. Now, for an arbitrary $n$, we can always find a pair of nodes $l,m$ in $F_n$ which saturates the inequality, with $\delta_{l,m}=n$. For example, in $F_2$ we can set $l$ equal to node $1$, and $m$ equal to note $3$, and for $F_3$ we have $(l,m)=(1,5)$ and for $F_4$ we have $(l,m)=(2,12)$. To construct an algorithm that provides $l$ and $m$ in the general case, we start with $F_3$, pictured in the top left panel of Fig.~\ref{SCFBdiameter2}, along with nodes $l$ and $m$, with a shortest path (which is not unique) between them coloured in red. We move to the top right panel, where we have $F_4$ overlaid with the additional edges highlighted with dotted lines. We can move $m$ to $m'$, and the length of the shortest path between $l$ and $m'$ is increased by 1. This is because the new $2$-simplex which we moved $m$ in to is \textit{not} glued to an edge which is a member of a shortest path between $l$ and $m$. We can repeat this process when moving from $F_4$ (bottom left panel) to $F_5$ (bottom right panel), however instead of moving $m'$, we move $l$ in a similar fashion. This is because the new simplices that are glued to the edges of $m'$ are a member of a shortest path between $l$ and $m'$, hence if we were to again move $m'$ to the new $2$-simplex joined to it, we would not increase the length of the shortest path. But if we move $l$ to $l'$ we again increase the length of the shortest path between our nodes by $1$. Repeating this process (by induction, using $F_3$ as our base case), alternating the movement of $l$ and $m$, we can find a shortest path between any two nodes of $F_n$, with length $n$.

Hence we can always find a $l$ and $m$ to give $\delta_{ij}$, and combining this with the bound $\delta_{ij}\leq n$ we have that $\max_{i,j} \{ \delta_{i,j} \}=n$, which concludes the proof.\end{proof}

\begin{figure}
\begin{centering}
\includegraphics[scale=1.5]{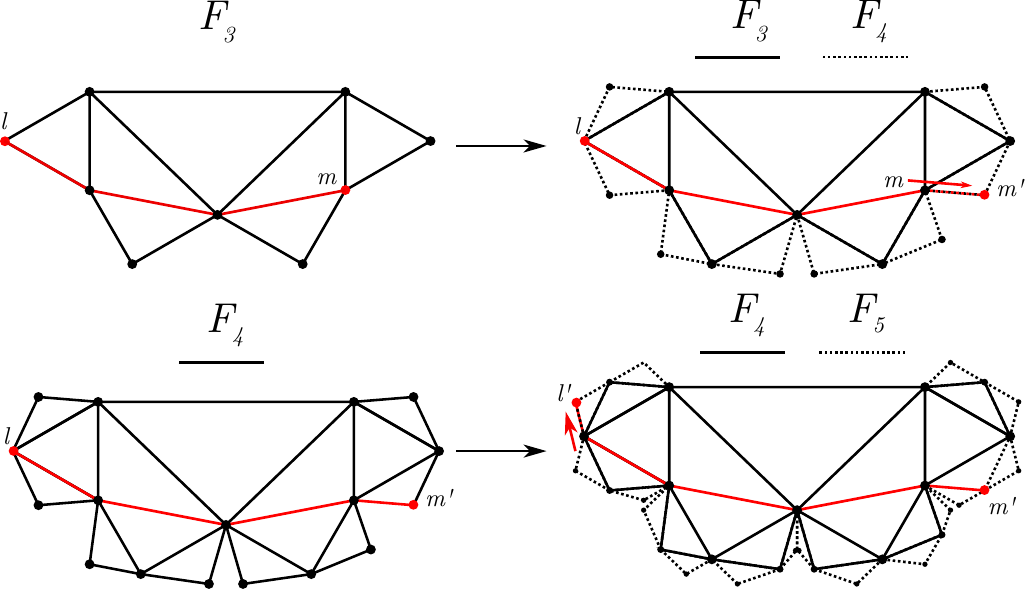}
\par\end{centering}
\protect\caption{\label{SCFBdiameter2}A visualisation of the algorithm used to create a shortest path $F_n$, with length $n$. In the top left panel we have $F_3$ with a shortest path between nodes $l$ and $m$ highlighted in red. In the top right panel we show how a new path can be created, in $F_4$, with length $4$, by moving $m$. In the bottom panels we show how we can move $l$ to create shortest path in $F_5$ with length $5$. The algorithm is described in the text.}
\end{figure}

%Lucas's method below, along with the figure I drew, commented out

%Now, for an arbitrary $n$, we can always find a pair of nodes $l,m$ in $F_n$ which saturates the inequality, with $\delta_{l,m}=n$. For example, in $F_2$ we can set $l$ equal to node $1$, and $m$ equal to note $3$, and for $F_3$ we have $(l,m)=(1,5)$ and for $F_4$ we have $(l,m)=(2,12)$. In general, a recipe for finding one particular pair $(l,m)$ for which $\delta_{lm}=n$ is the following: first, choose $l \in [0,R_{n-2}]$ and $m$ in the right-hand part of $[R_{n-1},R_n]$, i.e., $m\in[\tilde{R}_{n-2},R_n]$, where $\tilde{R}_{n-2}$ is the symmetric node of $R_{n-2}$ (see figure \ref{fig:ExtremalPoint}). This process is then self-similarly repeated for both $l$ and $m$ inside $F_{n-2}$  until reaching $F_1$. Each time the process is repeated, we increase the distance of the shortest path by $1$. 

%\lucas{Lucas: Ryan, please think if one can state the construction of $l$ and $m$ is a more easy way, and if so change the text accordingly}\\
%Hence the bound on $\delta_{ij}$ is tight and $\max_{i,j} \{ \delta_{i,j} \}=n$, what concludes the proof.\end{proof}

%\begin{figure}
%\begin{centering}
%\includegraphics[scale=1.3]{ExtremalPoint.pdf}
%\par\end{centering}
%\protect\caption{\label{fig:ExtremalPoint}}
%\end{figure}

According to the theorem above, we conclude $q({\bf A}_n)\ge n+1$. To evaluate how tight this bound is, in Fig \ref{fig:EigenvaluesSemilog} we plot the entire (point) spectrum of $F_n$ for $n\leq 10$ in semi-log. We can make several observations. First, the bound on $q({\bf A}_n)$ provided above does not seem to be tight, when comparing to the numerical evidence. On the contrary, the numerical evidence suggests instead that $q({\bf A}_n) = V_n=2^n+1$, i.e. all eigenvalues seem to be {\it distinct}, something that we leave as a conjecture.

Moreoever, the spectrum appears to be converging to a particular shape as $n$ increases. We will explore this fact further in \cref{sec:altogether}, but at this point we shall remark that the fact that the point spectra of ${\bf A}_n$ and ${\bf A}_{n-1}$ have resemblances is reminiscent of Cauchy's interlacing theorem~\cite{hwang2004cauchy}. Also, the spectrum is not symmetric and in particular the largest ($\lambda_{\text{max}}$) and smallest ($\lambda_{\text{min}}$) eigenvalues are different in modulus (thanks to the Perron-Frobenius theorem for primitive matrices, as discussed in \cref{sec:gelfand}). 
%z\textcolor{red}{This asymmetry can be exploited to obtain a bound for the chromatic number of the graph $\chi (F_n)\geq 1+ \frac{\lambda_{\text{min}}}{\lambda_{\text{max}}}$} \enzo{are we sure about the bound above? $\lambda_{\rm min}$ is always   negative, so the bound seems trivial...}\ryan{ryan:I think we can just remove the part in red, even if the bound is not trivial. We've not introduced the chromatic number yet, and I don't think we should}

\begin{figure}
\begin{centering}
\includegraphics[scale=0.5]{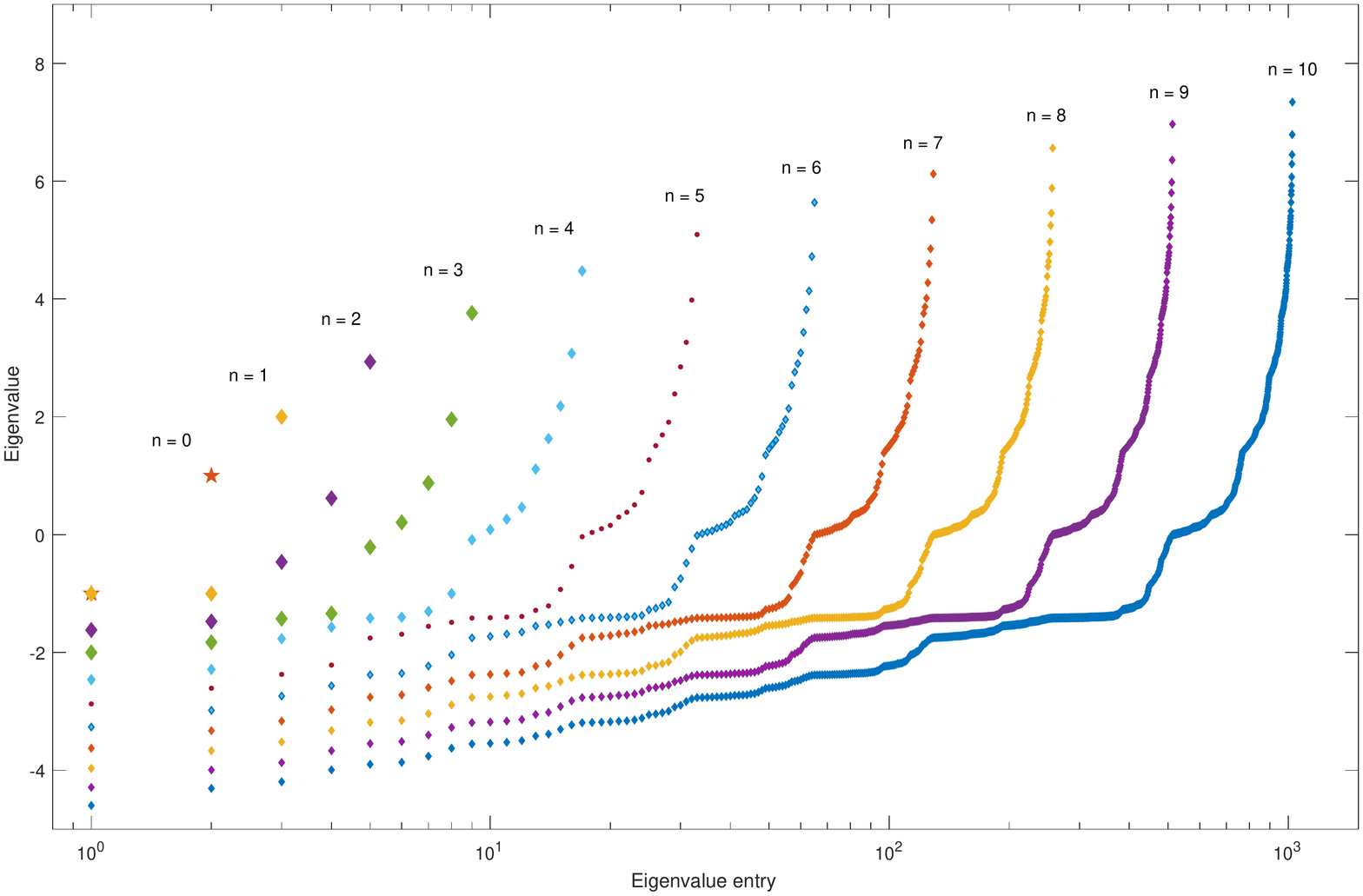}
\par\end{centering}
\protect\caption{\label{fig:EigenvaluesSemilog}A semilog plot of the spectrum (sorted in to ascending order) of $F_n$ for $0\leq n \leq 10$. We notice the spectrum seems to be converging to a particular curve and conjecture that all eigenvalues are distinct.}
\end{figure}

\subsection{\label{subsec:LargestEigenvalue}Largest eigenvalue  $\lambda_{\text{max}}$ for $F_n$}
\label{sec:n}
Here we continue to focus on the case $k=1$, and turn our attention to the largest eigenvalue of ${\bf A}_n$. The eigenvalues of ${\bf A}_n$ may be ordered as \[ \lambda_1\geq\lambda_2\geq\lambda_3\ldots\geq\lambda_{2^n+1} \] and as ${\bf A}_n$ is irreducible (the graph $F_n$ is undirected and connected), according to the Perron-Frobenius for non-negative irreducible matrices, $\lambda_1$ has multiplicity 1, we define $\lambda_\text{max}(F_n) = \lambda_1({\bf A}_n)$ (and similarly we define $\lambda_\text{min}(F_n) = \lambda_{2^n+1}({\bf A}_n)$). 
%\enzo{I AM NOT SURE ABOUT   THIS DEFINITION OF $\lambda_{\rm min}$. DID YOU MEAN $\lambda_{\rm     min}=\lambda_{2^n + 1}$ ???}\ryan{ryan: probably a typo, I have changed it now}.
Using the \texttt{eigs} function in {\sc Matlab} it is possible to efficiently calculate the largest eigenvalue of sparse matrices, even if the matrices are large. In figure \ref{fig:loglogfit} we plot, in a log-log scale, $\lambda_{\text{max}}(F_n)$ for $1\leq n \leq 26$.  The data fit very well to the power law dependence $\lambda_{\text{max}}\sim n^{\alpha}$ with $\alpha\approx 0.5$. 
%\enzo{is this the red line in   Fig.~\ref{fig:loglogfit}??? If so, we should be sure that the two   exponents are in agreement (in the figure we say that the exponent   is 0.5...)}\ryan{ryan:I think Lucas should check this, as I think he calculated the exponent with his graphing program. It's possible that it's more accurate now as we added some more values, and the exponent of 0.5 is the newer one.}

\begin{figure}
\begin{centering}
\includegraphics[scale=0.55]{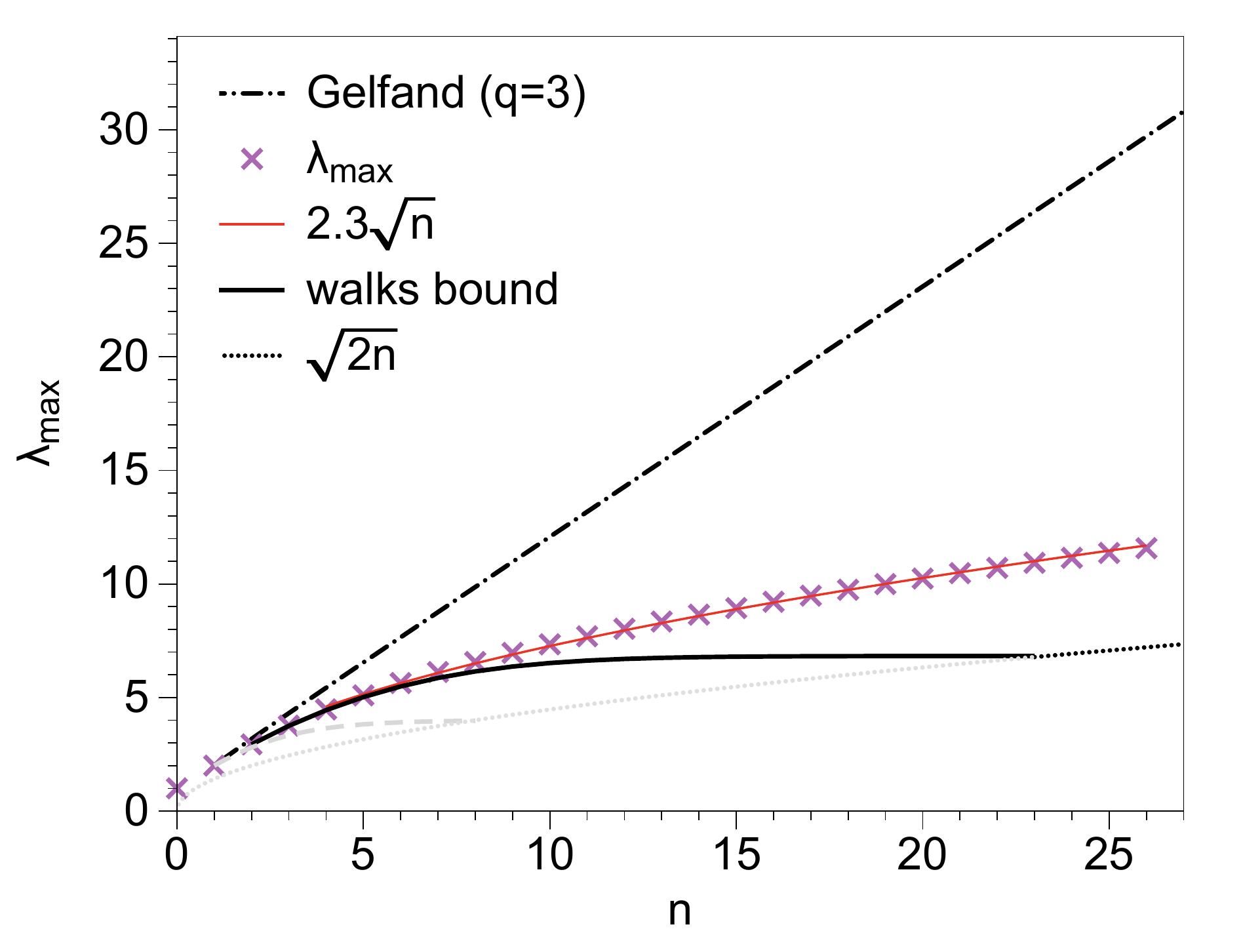}
\par\end{centering}
\protect\caption{\label{fig:loglogfit} Crosses denote numerical computation $\lambda_{\text{max}}(F_n)$, for $n=1,\dots,26$. The red solid line provides a power law fit $\lambda_{\text{max}} \sim n^{1/2}$. The rest of the lines are different analytical upper and lower bounds (see the text).}
\end{figure}

In this section our aim is to explain this scaling by finding adequate bounds.\\

%\begin{table}[]
%\centering
%\begin{tabular}{lc|clclclcl}
%\hline
%$n$  & $||{\bf A}_n||_{\infty}$ & $||{\bf A}_n^2||_{\infty}$ & $||{\bf A}_n^3||_{\infty}$           & $||{\bf A}_n^4||_{\infty}$ \\ \hline
%2    & 4             & 10            & 32                       & 90            \\
%3    & 6             & 20            & 80                       & 292           \\
%4    & 8             & 34            & 160                      & 706           \\
%5    & 10            & 52            & 280                      & 1424          \\
%6    & 12            & 74            & 448                      & 2558          \\
%7    & 14            & 100           & 672                      & 4236          \\
%8    & 16            & 130           & 960                      & 6602          \\
%9   & 18            & 164           & 1320                     & 9816          \\
%10   & 20            & 202           & 1760                     & 14054         \\
%11   & 22            & 244           & 2288                     & 19508         \\
%12   & 24            & 290           & 2912                     & 26386         \\ \hline
%\textsc{oeis} & $2n$          & $2n^2+2$      & $\frac{4}{3}n(n+1)(n+2)$ & n/a\\
%\hline   
%\end{tabular}
%\caption{The first 4 estimates of the upper bound of the spectral radius using Gelfand's formula for $k=1, 2, 3, 4$ without taking the $k$th root.}
%\label{tab:OEIS}
%\end{table}

\subsubsection{Gelfand's formula}
\label{sec:gelfand}
Gelfand's formula provides a bound for the spectral radius of an adjacency matrix $A$: $$\rho(A)=\lim_{q\to \infty} ||A^q||^{1/q},$$ where $||\cdot||$ is any matrix norm. In particular, for any finite $q\in\mathbb{Z^+}$ we have that $\rho(A)\leq||A^q||^{1/q}$.  

It is easy to prove that $\rho({\bf A}_n) = \lambda_{\rm max} (F_n)$. In fact, ${\bf A}_n$ is non-negative and irreducible, since the graph $F_{n}$ is non-empty, undirected and connected, and is also aperiodic, since each node of the graph belongs to at least one triangle. Consequently, ${\bf A}_n$ is primitive. The Perron-Frobenius theorem for non-negative primitive matrices guarantees that the largest eigenvalue is real, simple, and equal to the spectral radius $\rho(A)$. Therefore we can write $$\rho({\bf A}_n) = \lambda_\text{max}\leq||A^q||^{1/q}.$$ For simplicity, we choose $||\cdot||_{\infty}$, defined as $||A||_{\infty} = \max \limits _{1   \leq i \leq m} \sum _{j=1} ^n | a_{ij} |$.  We have that $||({\bf A}_n)^1||=2n$, this is because the node with the largest amount of 1-walks is the node with the largest degree; this is the central node and has degree $2n$ (see Prop.~\ref{prop:basic} below). For $q=2$ we have $\Vert({\bf A}_n)^2\Vert=2n^2+2$, a result which we prove in Appendix.~\ref{app:walks}. For $q=3$ we calculate $\Vert({\bf A}_n)^3\Vert$ for several values of $n$ and numerically find that they exactly fit a cubic equation:
\[
\Vert({\bf A}_n)^3\Vert=\frac{4}{3}n(n+1)(n+2)
\]
We did not find a closed formula for $4\leq q \leq 10$ (although one may very well exist). Taking respectively the 1st, 2nd and 3rd roots of these three formulas, we find that for $q<4$ the approximant to the spectral radius is essentially linear on $n$, providing our first estimated upper bound for $\lambda_\text{max}$. Because we did not find a closed expression for $4\leq q \leq 10$, we take $q=3$ as our `Gelfand's estimate' and we have a conjecture:
\begin{equation}
\lambda_\text{max} \leq \bigg [ \frac{4}{3}n(n+1)(n+2)\bigg]^{1/3}.
\label{gelfand}
\end{equation}
\subsubsection{Bounds on largest eigenvalue based on degree. }In order to improve the bound provided by Gelfand's formula, we now turn to the specific bounds for the largest eigenvalue that exist in the literature.  Some elementary bounds for the largest eigenvalue of a graph $G$ with maximum degree $d_{\text{max}}$ and average degree $\bar d$ \cite{das2004some} are:

\begin{eqnarray}
\label{bound_degree_eq1}
  && \max\{ \bar d, \sqrt{d_{max}}\} \leq \lambda_{\text{max}}\leq d_\text{max} \\
  \label{bound_degree_eq2}
  && \lambda_{\text{max}}\leq \max{\{ \sqrt{d_{i}d_{j}}: 1\leq i,j\leq n, v_{i}v_{j}\in E \}},
\end{eqnarray}
where $E$ is the edge set.
We apply these bounds to $F_n$. We summarise the bounds in the following proposition:
\begin{prop}
Consider $F_n$. Then
\begin{enumerate}[(a)]
\item The largest degree of $F_n$ is found in its central vertex and is $d_{\text{max}}(F_n)=2n$.
\item The vertices with second largest degree are the boundary ones (first and last) and each have degree $n+1$.
\item The average degree is $\bar d(F_n)=4-6/(2^n+1)$
\item $\max{\{ \sqrt{d_{i}d_{j}}: 1\leq i,j\leq n, v_{i}v_{j}\in E \}}=\sqrt{2n(n+1)}$
\end{enumerate}
\label{prop:basic}
\end{prop}
\begin{proof}
First, observe that for $n\geq1$ we have $F_n=F_{n-1}\otimes F_{n-1}$, and in the inflation process the only vertices whose degree increases are the border ones (leftmost and rightmost).\\
\noindent Proofs of (a) and (b) are then by induction on $n$:
For $p=1$ we have that $d_{\text{max}}(F_p)=2$, found in the central vertex, and similarly for the first and last vertex $d=2$ as well. Then,

 - Assume $d_{\text{max}}(F_p)=2p$. For $n=p+1$, by construction we have $F_{p+1}=F_{p}\otimes F_{p}$, so the only vertices that acquire new edges are at the borders of $F_p$. In particular, the central vertex in $F_{p+1}$ is the one acquiring more edges, and by construction this vertex is built merging the rightmost and leftmost vertex of $F_p$, hence the central vertex of $F_{p+1}$ has degree $2p+2p=2(p+1)$. This finishes the proof for (a).
 
 - Assume that the border vertices (leftmost and rightmost) in $F_p$ have degree $p+1$. In $F_{p+1}=F_{p}\otimes F_{p}$, inflation adds an additional edge between the leftmost vertex in the first copy of $F_p$ and the rightmost vertex in the second copy of $F_p$, and therefore the degree for these nodes in $F_{p+1}$ is just $p+1+1$, finishing the proof for (b).
 
Moreover, a proof for (c) directly follows from Proposition \ref{prop:initial} by remarking that $\bar d(F_n)=2 E_n/V_n$.

\noindent Finally, the vertices with largest degree in $F_n$ are the central vertex, with degree $2n$, and the leftmost and rightmost vertices, each of them having degree $n+1$ as previously proved. By construction the central vertex in $F_n$ is always linked with the leftmost and rightmost vertices, hence the identity (d) holds. 
\end{proof}

%\ryan{ryan: the proof of (a) is incorrect, it assumes that we know that the degree of the boundary nodes are 2p. So we could prove (b) first, but I don't think we need an induction; just state that when we do the inflation rule, the boundary nodes get an extra edge (then we can just say, by induction, this means that the degree of the boundary nodes is n+1, because for n=1 the boundary nodes have deg 2), then quickly state that (a) follows almost directly from (b) as the central node has degree twice that of the degree of the boundary nodes for $F_{n-1}$}

In summary, we find the following bounds based on the degree for $\lambda_{\text{max}}$:
%which is larger than $\sqrt{d_{\text{max}}}$ for $n<9$. 
%Thus $\lambda_{\text{max}} \leq \sqrt{2n(n+1)}$ which is indeed a better bound than $2n$.
\begin{eqnarray}
  && \lambda_{\text{max}} \leq \sqrt{2n(n+1)},\label{bound_degree0}\\ 
  && \lambda_{\text{max}}\geq \bar d(F_n)=4-6/(2^n+1), \ \text{for}\ n<9,\\
  && \lambda_{\text{max}}\geq \sqrt{2n}, \ \text{for}\ n\geq 9.
  \label{bound_degree}
\end{eqnarray}

Note that asymptotically the lower bound is already $\sim n^{1/2}$ and is therefore tight, whereas the upper bound $\sqrt{2n(n+1)}$ is still linear and worse than our estimate derived from Gelfand's formula.

\subsubsection{Bounds on largest eigenvalue based on walks. }\label{sec:walkbound} There exists a general bound for $\lambda_{\text{max}}$ based on number of walks on the graph up to order $3$. Let $a(n)$, $b(n)$, $c(n)$ and $d(n)$ be the total number of 3-walks, 2-walks, 1-walks and 0-walks respectively (observe that $d(n)$ is simply the number of vertices and $c(n)$ is just twice the number of edges). 
Then a lower bound is~\cite{EnzoBook}:
\begin{equation}
\lambda_{\text{max}} \geq \frac{d\cdot a - b\cdot c + \sqrt{d^2\cdot a^2 - 6abcd - 3c^2\cdot b^2 + 4(ac^3 + db^3)}}{2(bd - c^2)}.
\label{walks}
\end{equation}

In appendix~\ref{app:walks} we provide a proof for $b(n), c(n)$ and $d(n)$ along with an estimation for $a(n)$. According to these, we state that for $F_n$:
\begin{eqnarray}
&&a(n) = 160\cdot2^n - 4n^3 - 30n^2 - 104n - 158 \label{a}\\  %number of 3-walks for F_n
&&b(n) = 24\cdot 2^n - 2n^2 - 12n - 22 \nonumber \\    %number of 2 walks for F_n
&&c(n) = 4\cdot 2^n - 2  \nonumber \\ %number of 1-walks for F_n (= twice number of edges)
&&d(n) = 2^n + 1 \nonumber  %number of 0-walks for F_n (= number of nodes).
\end{eqnarray}

Based on the leading terms of $a,b,c$ and $d$ it is clear that the lower bound in Eq.\ref{walks} is asymptotically constant, and is therefore a very loose bound for large $n$. However the expression is extremely good for small values of $n$, as shown in Fig. \ref{fig:loglogfit}.

Summing up, we have exploited different properties such as spectral radius, degree and walks, and we have obtained several possible bounds accordingly (see Eqs. \ref{gelfand}--\ref{walks}). The best upper bound is the Gelfand estimate (Eq.\ref{gelfand}) which is nonetheless still a loose bound. On the other hand the best lower bound is given by the walks bound (Eq.\ref{walks}) for $n<24$ and by the degree bound (Eqs.\ref{bound_degree}) for $n\geq24$. The scaling of this latter bound seems to be tight. These bounds have been displayed, along with the numerical estimate of $\lambda_{\text{max}}$, in Fig. \ref{fig:loglogfit}.

\subsection{Other spectral properties of $F_n$: The Tree Number} 
The tree number of a graph $G$ is the total number of spanning trees, and we will denote it by $\kappa(G)$. To calculate $\kappa(F_n)$ we make use of Kirchhoff's theorem, or the matrix tree theorem:

\begin{thm}[Kirchhoff's theorem (The Matrix Tree Theorem); \cite{biggs1993algebraic}]
For a given connected graph $G$ with $n$ labeled vertices, let $\mu_1, \mu_2, \dots, \mu_{m-1}$ be the non-zero eigenvalues of its Laplacian matrix $\bf L = D - A$., where $\bf D$ is the degree matrix (a diagonal matrix with vertex degrees on the diagonals). Then the number of spanning trees of $G$ is given by
\begin{equation} 
\kappa(G)=\frac{\prod_{i=1}^m \mu_i}{m}
\end{equation}
\end{thm}

We have numerically computed $\kappa(F_n)$ for $n=1, \dots,4$, the results are shown in Table \ref{TreeNumbers}. Interestingly,  \textsc{oeis} states that this sequence (A144621) corresponds to the number of oriented spanning forests of the regular ternary tree with depth $n$ that are rooted at the boundary (i.e., all oriented paths end either at a leaf or at the root), which is given by the recurrence
\begin{equation}\label{alphan}
\alpha_{n+2}=\alpha_{n+1}(3\alpha_{n+1}-2{\alpha^2_n}) \quad  \text{where} \quad \alpha_0=0,\ \alpha_1=1. 
\end{equation}
Hence we conjecture that $\kappa(F_n)=\alpha_{n+1}$.

%For a given $n$, $\kappa(F_{n}^k)$ grows extremely quickly with $k$ and for $n=1, 2, 3$ the sequence in $k$ of $\kappa(F_{n}^k)$ does not correspond to any sequences in \textsc{oeis}.

\begin{table}[]
\centering
\begin{tabular}{|c|c|}
\hline
$n$ & $\kappa(F_n)$   \\ \cline{1-2}
1   & 3                \\
2   & 21               \\
3   & 945              \\
4   & 1845585 \\
\hline      
\end{tabular}
\caption{Tree numbers for $F_n$}
\label{TreeNumbers}
\end{table}

\subsection{Largest eigenvalues $\lambda_{\text{max}}$ for $F_n^k$}
\label{sec:k}
In this section we start by focusing on the set of graphs generated via the concatenation rule $\oplus$ as introduced in definition~\ref{concatenation}. We initially are interested in exploring the role of $k$ in the largest eigenvalue of the adjacency matrix.  To begin with, we fix $n$ and explore (numerically) how $\lambda_{\text{max}}(F_n^k)$ changes as we increase $k$. We have calculated $\lambda_{\text{max}}(F_n^k)$ for $1\leq n \leq 4$ and $1 \leq k \leq10$. 
%\footnote{These values are chosen due to computational constraints} \ryan{ryan: I think I calculated this right at the beginning, I can now calculate $lambda_{\text{max}}$ for larger $n$ and $k$ if need, either way I don't think we need the footnote}. 
Results are shown in Fig.~\ref{fig:LarEigF_K}. After a transient growth, we notice that for each $n$ the $\lambda_{\text{max}}$ appears to converge to a finite value as $k$ increases. This observation can be made rigorous: 
\begin{thm}\label{thm:lmax}
Let $n>0$ be fixed and consider the graph $F_n^k$ as $k$ increases. Then the largest eigenvalue of its adjacency matrix converges as $k\to\infty$.
\end{thm}
\begin{proof}
Recall that the largest eigenvalue is bounded by the largest degree of the graph hence in our case, $\lambda_{\text{max}}(F_n^k)\leq d_{\text{max}}$.  Now, the node with the largest degree in $F_n^2$ is the central node, which by construction inherits the edges from the left and right boundary nodes in $F_n$. These boundary nodes have degree $n+1$, hence
$$d_{\text{max}}(F_{n}^2\equiv F_{n}\oplus F_{n})= 2(n+1).$$ 
Adding additional copies of $F_n$ does not change the maximum degree, because only one of the boundary nodes in $F_n^k$ will have their degree increased from $n+1$ to $2(n+1)$. In other words, the node with largest degree is maintained constant as new motifs are concatenated. Therefore $\lambda_{\text{max}}(F_{n})$ is bounded from above. Furthermore, as a consequence of Cauchy's Interlacing Theorem we have that  $\lambda_{\text{max}}(F_{n}^{k+1}) \geq \lambda_{\text{max}}(F_n^k)$. Therefore $\lambda_{\text{max}}(F_n^k)$ is an increasing sequence in $k$, bounded above, hence converges.
%\ryan{ryan: is CIT overkill here? is there a simpler justification for the above inequality?} 
\end{proof}
%\ryan{ryan: also the proof above is a bit long. I cleaned it up a bit, but we can probably just say "adding copies of $F_n$ with the concatenation rule does not change the maximum degree, hence it's a sequence which is bdd above. It's also increasing because of...}

\begin{figure}
\begin{centering}
\includegraphics[scale=0.5]{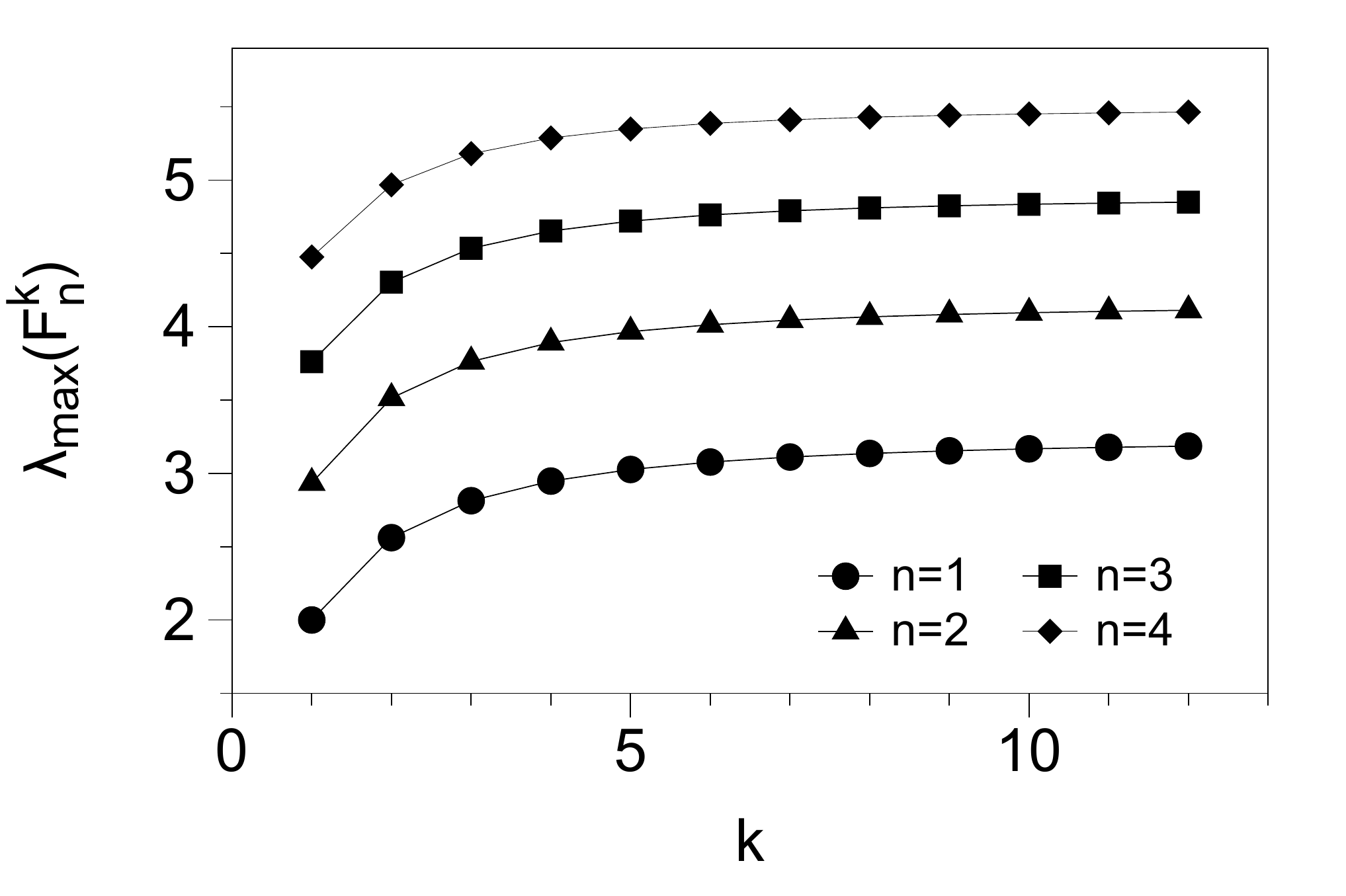}
\par\end{centering}
\protect\caption{\label{fig:LarEigF_K} Plot of $\lambda_{\text{max}}(F_n^k)$ for fixed values of $n$ as a function of $k$. In each case the largest eigenvalue converges to a value independent of $k$, a result proved in theorem~\ref{thm:lmax}.}
\end{figure}

We have now understood the dependence of $\lambda_{\text{max}}(F_n^k)$ on $k$, and we are now in a position to discuss a general expression for the largest eigenvalue in the general case of $F_n^k$. This is a two-parameter discrete function $$\lambda_{\text{max}}(F_n^k): n,k \in \mathbb{N^+}\to \mathbb{R}^+ $$

%In figure \textcolor{red}{FIGURE} we plot a numerical estimate of $\lambda_\text{max}(n,k)$ in the range \textcolor{red}{RANGE}. Analytical bounds for $\lambda_\text{max}(n,k)$ can be obtained again using the bounds based on degrees and walks. 

We summarise the bounds for $\lambda_{\text{max}}$ for general $n$ and $k$ in the following proposition:
\begin{prop}
Consider the graph $F_n^k$, where $n\geq 0, k\geq2$ (the case $k=1$ reduces to $F_n$). Then the following hold:
\begin{enumerate}[(a)]
\item $d_{\text{max}}(F_n^k)=2(n+1)$ (independent of $k$).
\item $\displaystyle {\bar d}(F_n^k)=\frac{2E_n^k}{V_n^k}=\frac{2k(2^{n+1}-1)}{k2^{n}+1}$  (asymptotically independent of $k$).
\item $\max{\{ \sqrt{d_{i}d_{j}}: 1\leq i,i\leq n, v_{i}v_{j}\in E \}}=2n+2$ (independent of $k$).
\end{enumerate}
\end{prop}
\begin{proof}
Proposition (a) comes from Theorem~\ref{thm:lmax}\\
Proposition (b) comes from Prop.~\ref{prop:initial} along with the fact that the average degree of a graph is twice the number of edges divided by the number of nodes\\
Proposition (c) is trivially proved by observing that for $k\geq2$, $F_n^k$ has multiple nodes with maximum degree; these nodes are always connected and have degree $2(n+1)$.
\end{proof}
These results provide the following bounds, which are equivalent to the bounds to Eqs.\ref{bound_degree0}--\ref{bound_degree} but in the general case where $k\geq 2$:
\begin{eqnarray}
  && \lambda_{\text{max}}(F_n^k) \leq 2(n+1),\\
  && \lambda_{\text{max}}\geq \bar d(F_n^k)=\frac{2k(2^{n+1}-1)}{k2^{n}+1}, \ \text{if}\ n<7,\\
  && \lambda_{\text{max}}\geq \sqrt{2(n+1)}, \ \text{if}\ n\geq 7.
  \label{bound_degree_general}
\end{eqnarray}
Additionally, we were able to estimate the number of walks (as explained in Appendix~\ref{app:walks}) in the case of $F_n^k$, yielding:
\begin{eqnarray}
&&a(n,k) = (320\cdot2^n - 4n^3 - 48n^2 - 196n - 312) + (k - 2)(160\cdot2^n - 16n^2 - 88n - 152) \label{a_general}\nonumber \\  %number of 3-walks for F_nk
&&b(n,k) = 24\cdot 2^n - 2n^2 - 12n - 22 + (k - 1)(24\cdot2^n - 8n - 20)\nonumber \\    %number of 2 walks for F_nk
&&c(n,k) = 4k\cdot2^n - 2k  \nonumber \\ %number of 1-walks for F_nk (= twice number of edges)
&&d(n,k) = k\cdot2^n + 1 \nonumber  %number of 0-walks for F_nk (= number of nodes)
\end{eqnarray}
%The best bounds are plotted together with the actual $\lambda_\text{max}(n,k)$ in figure \textcolor{red}{FIGURE}.

\subsection{The complete spectrum of $F_n^k$: tridiagonal $n$-block Toeplitz matrices.}
\label{sec:altogether}

%\textcolor{red}{THE FIRST PARAGRAPH COMES FROM THE PRECEDING SECTION BUT I MOVED IT HERE BECAUSE IT TALKS ABOUT THE COMPLETE SPECTRUM, NOT LAMBDA MAX. ACTUALLY SOMEONE CAN CHECK THIS PARAGRAPH?}
%Now that the properties of $\lambda_{\text{max}}$ have been discussed in detail, we now turn our attention to the complete spectrum of $F_n^k$. A first observation is that, due to the Perron-Frobenius theorem \textcolor{blue}{CHECK THIS} we have that for all $k$,  $|\lambda_{\text{max}}(F_n^k)|>|\lambda_\text{min}(F_n^k)|$ hence $\lambda_\text{min}(F_n^k)$ forms a decreasing sequence in $k$, bounded below, hence converges. Since $\lambda_\text{max}(F_n^k)$ also converges with $k$ for a fixed $n$, we can therefore expect the spectrum of $F_n^\infty$ (recall that the graphs are pictured in Figure \ref{fig:Single-motifs-}) to coincide with a closed interval of the real numbers, bounded by the two limits of the sequences of maximum and minimum eigenvalues.\\

%\ryan{ryan: I have commented out the above paragraph (see .tex). Instead of the decreasing sequence in $k$ (which might be incorrect) we can just state the theorem that for a locally finite infinite graph, the spectrum is a subset of the interval $[-\lambda_{\text{max}},\lambda_{\text{max}}]$ instead}
%\ryan{ryan: we can start this section with the paragraph below then add the paragraph above later on, or some details from it anyway}

We now turn our attention to the adjacency matrices of $F_n^k$ and their particular form. As a preamble, observe that the concatenation operation $\oplus$ that generates $F_n^{k}$ from $F_n$ is in some sense `close' to a direct sum. We recall that the direct sum of a matrix {\bf A} with itself is the matrix formed by placing {\bf A} as two non-overlapping diagonal blocks. The eigenvalues of the direct sum of two copies of the same matrix  {\bf A} are just the eigenvalues of {\bf A} (with twice the multiplicity in each case). If we `approximate' $\oplus$ as just being the direct sum operation, then trivially the eigenvalues of $F_n^k$ would be the same as the eigenvalues of $F_n$. In particular, $\lambda_{\text{max}}$ would be fully independent of $k$. Of course, $\oplus$ is not a direct sum, however $\lambda_{\text{max}}(F_n^k)$ is independent of $k$ in the limit $k\to\infty$. With a bit of hand-waving, we could say that the larger $n$, the `closer' $\oplus$ is to a direct sum and therefore the more independent the spectrum is from $k$.

We start now our analysis by fixing $n$ and letting $k$ increase. For $n=0$, $F_0^k$ is trivially a 2-regular chain whose adjacency matrix $\textbf{A}_0^k$ whose structure is tridiagonal Toeplitz:

\[
\textbf{A}_0^k= 
\begin{bmatrix}
0    & 1        & 0                  &             &0            & 0  \\
1    & 0        & 1                 & \ddots  &             & 0 \\
0   & 1          & \ddots     & \ddots   & \ddots &    \\
      & \ddots & \ddots     & \ddots &1        & 0  \\
0    &             & \ddots    &1                 & 0       & 1 \\
0  & 0           &                   & 0            & 1      & 0\\
\end{bmatrix}
\]

Accordingly, through direct calculation of the adjacency matrix, we can express the spectrum in closed form $$\text{spec}(F_0^k)=\{2\cos \frac{h\pi}{k+2}, \ h=1,\dots,k+1\}.$$
A plot of this spectrum for $k=2^{10}$ is shown in \ref{fig:EigF_0}. In particular, as $\cos(x)$ monotonically decreases on $[0,\pi]$, the largest value is found for $h=1$ and thus $$\lambda_{\text{max}}(F_0^k)=2\cos \frac{\pi}{k+2},$$
hence we have that $\lim_{k\to\infty}\lambda_{\text{max}}(F_0^{k})=2$. 
%Indeed, $F_0^{\infty}$ is the so called two way-infinite path which we can identify with $\mathbb{Z}$. 
%It is well known that $\text{spec}(\textbf{A}_n^\infty)=[-2,2]$~\cite{spectral_infinite}. 

\begin{figure}
\begin{centering}
\includegraphics[scale=0.7]{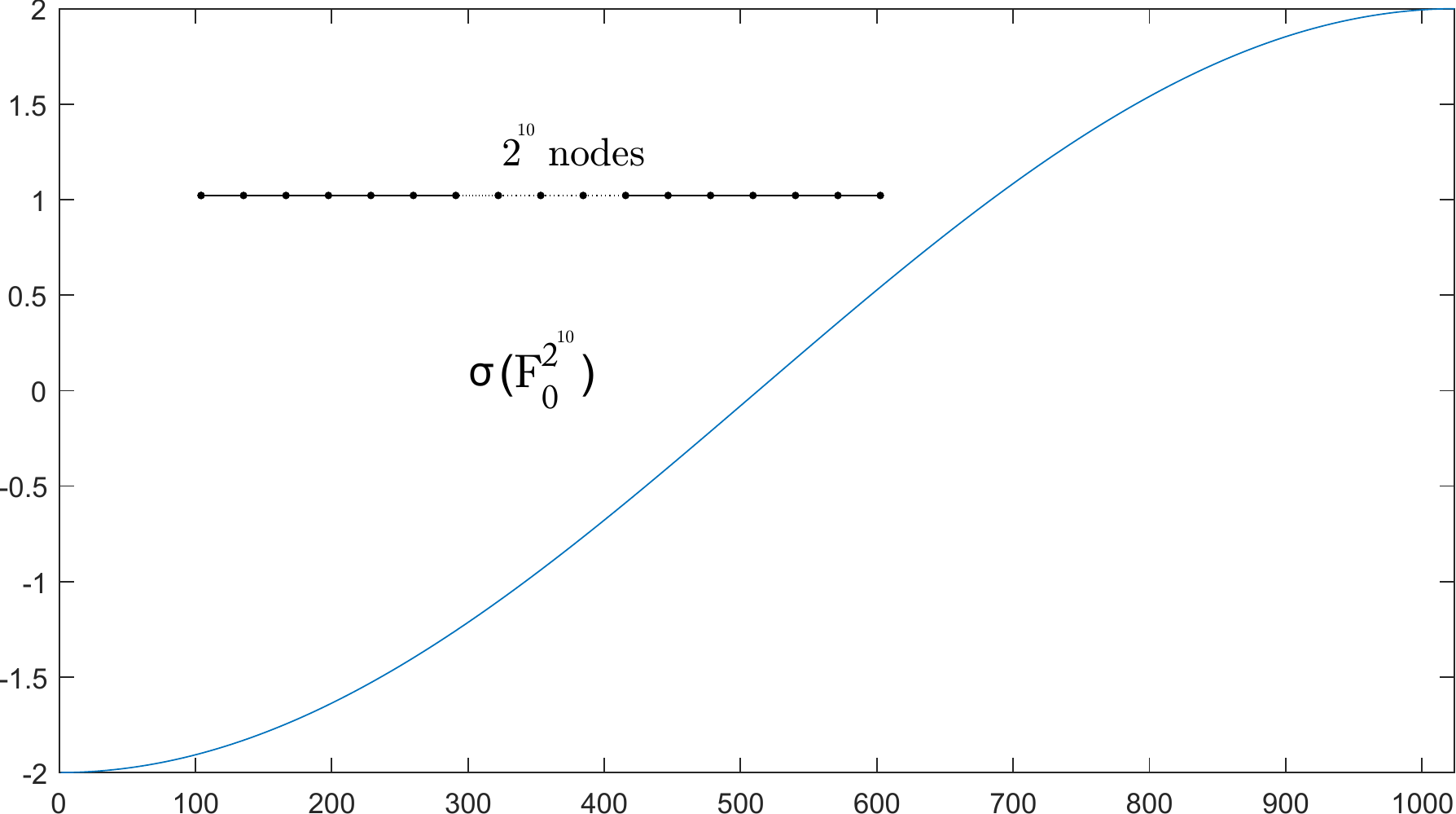}
\par\end{centering}
\protect\caption{\label{fig:EigF_0}The spectrum of $F_0^{2^{10}}$.}
\end{figure}

\begin{figure}
\begin{centering}
\includegraphics[scale=0.7]{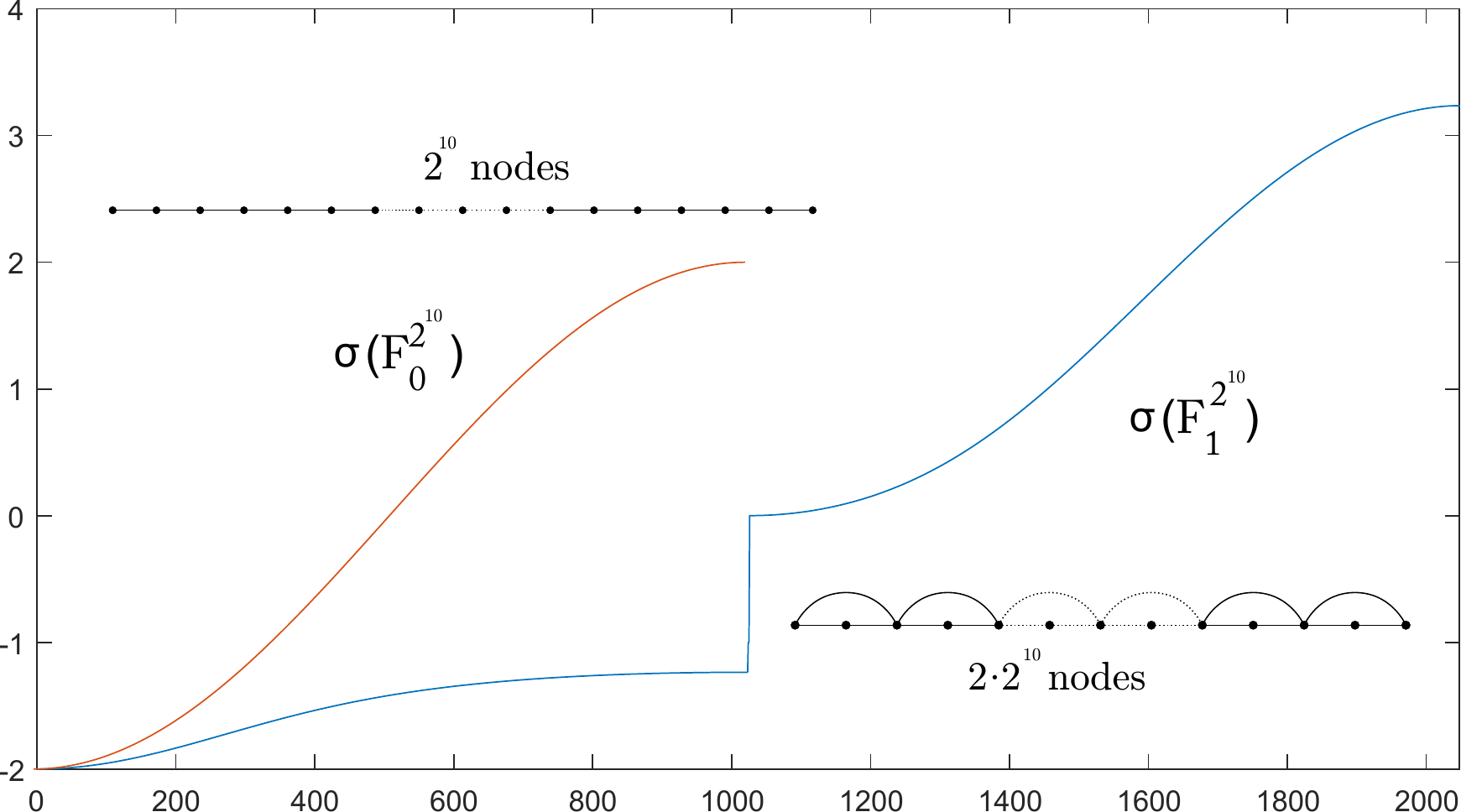}
\par\end{centering}
\protect\caption{\label{fig:EigF_0F_1}The spectrum of $F_0^{2^{10}}$ and $F_1^{2^{10}}$.}
\end{figure}

For $n=1$, $F_1^k$, the adjacency matrix is no longer tridiagonal Toeplitz anymore, however it can be expressed as a tridiagonal \textit{block} Toeplitz matrix of the shape

\[ \left( \begin{array}{ccccc}
a & b & 0 & \dots & \\
b^T & a & b & 0 & \dots \\
0 & b^T & a & b & 0 \end{array} \right)\] 
where 
\[ a=\left( \begin{array}{cc}
0 & 1 \\
1 & 0
\end{array} \right),\ 
b=\left( \begin{array}{cc}
1 & 0 \\
1 & 0
\end{array} \right),\ 
b^T=\left( \begin{array}{cc}
1 & 1 \\
0 & 0
\end{array} \right).
\]

%\ryan{ryan: consider not defining c, and in the big matrix just putting $a$, $b$, and $b^t$, it might read better here}
This is a special type of tridiagonal block Toeplitz matrix. In general, if we look at the adjacency matrix ${\bf A}_n^k$ associated to $F_n^k$, there exists a self-similar process underlying the construction of ${\bf A}_n^k$ in terms of ${\bf A}_{n-1}^k$. For instance, ${\bf A}_0^k$  is just a tridiagonal Toeplitz matrix with null diagonal elements. Now, ${\bf A}_1^k$ is not tridiagonal nor Toeplitz anymore as we have seen, but we recover a tridiagonal Toeplitz shape if we consider blocks $2\times 2$ as the elements of this new matrix, or equivalently ${\bf A}_1^k$ is a tridiagonal block Toeplitz matrix. Similarly, ${\bf A}_2^k$ is no longer a tridiagonal block Toeplitz matrix, but if we consider that the elements of ${\bf A}_2^k$ are blocks of blocks ($2\times 2$ matrices whose elements are in turn blocks), then in the structure of ${\bf A}_2^k$ is again tridiagonal Toeplitz (we may call it tridiagonal superblock, or 2-block Toeplitz). For instance, the structure of ${\bf A}_2^k$ can be expressed as
\[ \left( \begin{array}{ccccc}
A & B & 0 & \dots & \\
B^T & A & B & 0 & \dots \\
0 & B^T & A & B & 0 \end{array} \right)\] 
where 
\[ A=\left( \begin{array}{cc}
a & b \\
b^T & a
\end{array} \right),\ 
B=\left( \begin{array}{cc}
c & 0 \\
b & 0
\end{array} \right),\ 
B^T=\left( \begin{array}{cc}
c & b^T \\
0 & 0
\end{array} \right),\ 
c=\left( \begin{array}{cc}
1 & 0 \\
0 & 0
\end{array} \right)
\]

This process can be applied iteratively and hence we can show that  ${\bf A}_n^k$ has a tridiagonal $n$-block Toeplitz structure. In this case, an $n$-block is equivalent to a $2^n\times 2^n$ block. In other words, a tridiagonal $n$-block Toeplitz matrix is equivalent to a tridiagonal block Toeplitz matrix where each block is indeed a $2^n\times 2^n$ matrix. We weren't able to find such shape in the literature but we speculate that the set of symmetries present in the recursive bulding of  ${\bf A}_n^k$ could be exploited to extract properties about its spectrum. Additionally, in Figure \ref{fig:EigF_0F_1} we plot the spectrum of $F_0^{2^{10}}$ and $F_1^{2^{10}}$. For the spectrum of $F_1^{2^{10}}$, we notice two distinct curves separated by a discontinuity, with a length of approximately $2^{10}$ each. Each of the two curves look appropriate rescalings of the spectrum of $F_0^{2^{10}}$. The same pattern can be observed in $F_n^{2^{10}}$ for $n=2, 3, 4$ with the spectrum of $F_n^{2^{10}}$ having $2^n$ distinct curves, each separated by a jump. We conjecture that for a fixed $k$,  the spectrum of $F_n^{k}$ consists of $2^n$ distinct curves.\\ 
%Furthermore, as we have shown in \cref{sec:k} that  $\lambda_{\text{max}}(F_n^k)$ converges with $k$, we hypothesize further that the smallest eigenvalue also converges, hence the spectrum of $F_n^{\infty}$ will coincide with a closed \textcolor{blue}{or open?} interval of the real numbers, the end points of which are these two limits.\\
Finally, in figure \ref{universal} we plot the complete point spectrum of $F_{10}$, and compare it with $F_n^k$ with the same number of nodes: $F_8^4$, $F_6^{16}$, $F_4^{64}$ and $F_2^{256}$. We can see how for small $n$ the distinct curves are very obviously separated by discontinuities, and these smear out as $n$ increases. The spectrum seems to converge to a somewhat universal shape. We conjecture that this self-similar process is reminiscent of the recursive way of building the $n$-block tridiagonal Toeplitz adacency matrices, and we leave this as an open problem.

\begin{figure}
\begin{centering}
\includegraphics[scale=0.45]{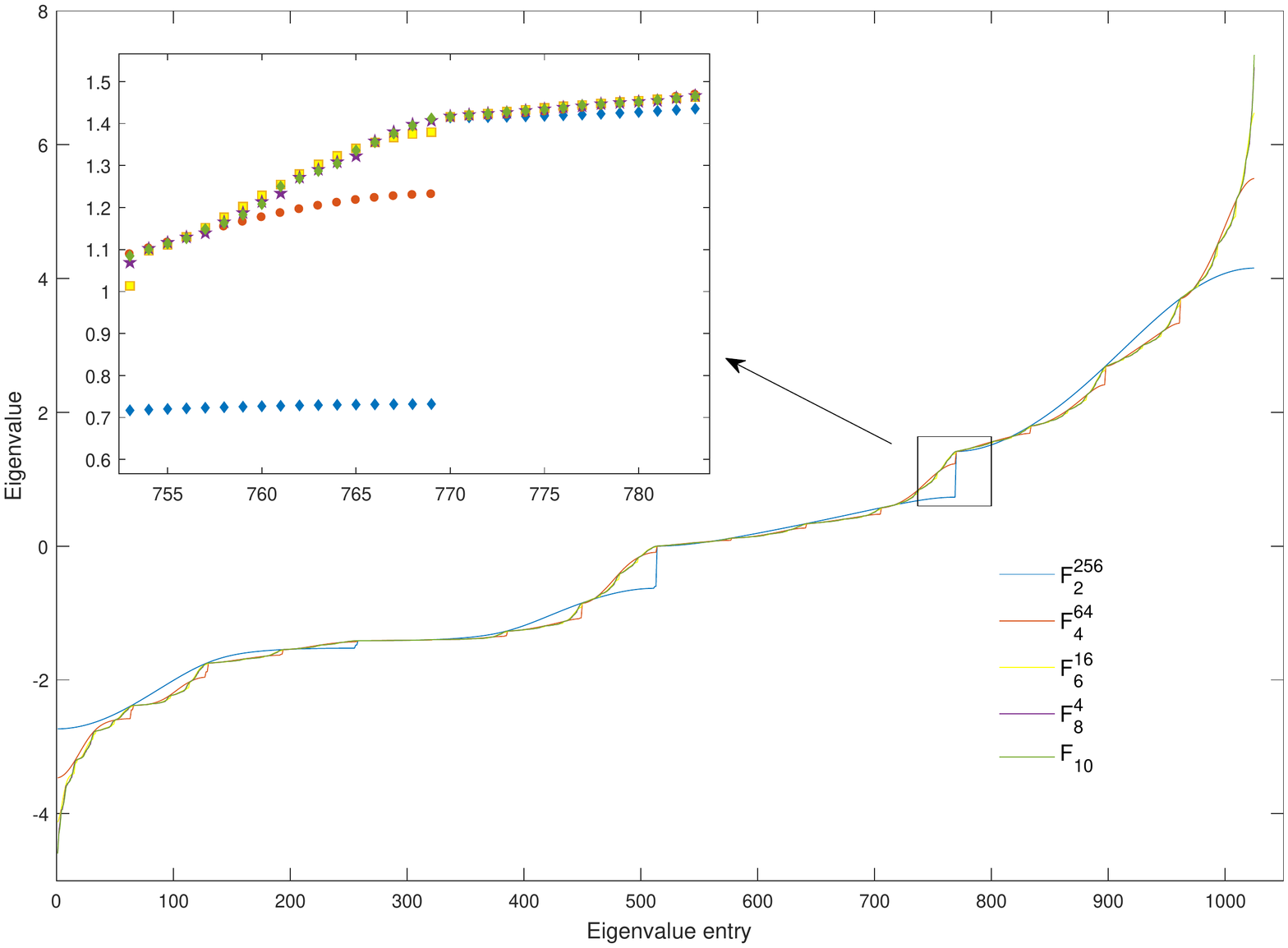}
\par\end{centering}
\protect\caption{Rescaled curves for $F_n^k$, where we show that, when comparing graphs with the same number of nodes, the spectrum collapses to a universal shape.}
\label{universal}
\end{figure}

\subsection{Determinant of Feigenbaum Graphs}
We close this section on the properties of $F_n^k$ by exploring the determinant of $F_n^k$, which is defined as the determinant of the adjacency matrix ${\bf A}_n^k$. We outline and prove the following theorem:

\begin{thm} \label{thm_ryan}
The determinants of $F_{n}^k$ satisfy

\[
det(F_{n})=\begin{cases}
1, & n=0\\
2, & n=1\\
-2, & \forall n\geq2.
\end{cases}
\]

Moreover, for $k\geq1$ and $n\geq2$ we have

\[
det(F_{n}^{k})=-2k
\]

\end{thm}
\begin{proof}
For $k=1$, we can directly calculate $det(F_{0})$ and $det(F_{1})$. To push beyond this is a little bit more difficult. It is too tricky to directly calculate the determinant of the adjacency matrices, despite them having a recursive form. We then follow a graph theoretical proof, for which we will have to state a definition and a well-known theorem, and then state and prove two lemmas.\\
We start by defining a {\it spanning elementary subgraph} \cite{biggs1993algebraic}:

\begin{defn} [Spanning Elementary Subgraph]
An elementary subgraph is a simple subgraph, each component of which is regular and has degree $1$ or $2$, i.e., each component is either a single edge or a cycle. A spanning elementary subgraph (S.E.S) of a graph $G$ is an elementary subgraph which contains all vertices of $G.$
\end{defn}

We now make use of the following Theorem:

\begin{thm}[Harary 1962; \cite{biggs1993algebraic}]
Let $A$ be the adjacency matrix of a graph $G$, let $v$ be the number of vertices, $e$ the number of edges and $l$ the number of components. Then 

\begin{equation}
det(A)=\sum_{H}(-1)^{r(H)}2^{c(H)}\label{eq:det}
\end{equation}

where the summation is over all spanning elementary subgraphs H of G, $r(H)=v-l$ is the rank of $H$ and $c(H)=e-v+l$ is the co-rank.
\end{thm}

We note that the co-rank of an elementary subgraph is just the number of cycles in the graph. Our task is thus to find all the spanning elementary subgraphs with their corresponding ranks and co-ranks. 

\begin{lem}
We have only two configurations for elementary subgraphs (for $n\geq3$). The first is just the cycle containing all vertices (shown as the "largest cycle" in Fig.~\ref{Cycles}) and the second is any other cycle from Fig.~\ref{Cycles} with the remaining nodes connected by single edges
\label{lema1}
\end{lem}

\begin{proof}
Without loss of generality we consider the structure of $F_{3}$ only. Because of the recursive property of the Feigenbaum graphs, all the arguments used here can be applied directly to any $F_{n}$ with $n\geq2$. 

\begin{figure}
\begin{centering}
\includegraphics[scale=0.4 ]{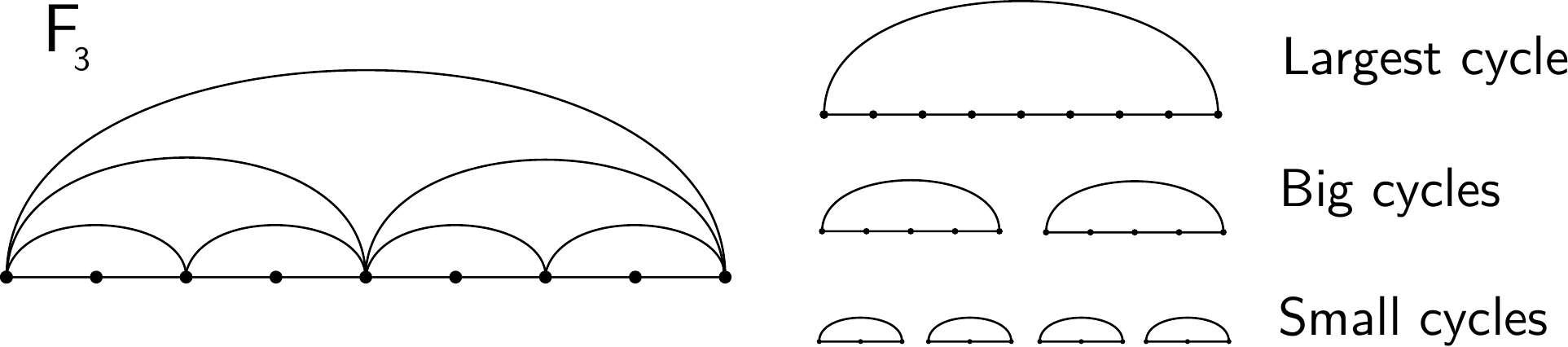}
\par\end{centering}
\protect\caption{\label{Cycles}Diagram showing $F_3$ with the relevant cycles. The largest cycle is taken as the Hamiltonian path along with the edge connecting the first and last vertices. The big and small cycles are taken by following the Hamiltonian path but taking an edge back to the starting node.}
\end{figure}

The largest cycle (created by taking only the outside edges of the outerplanar graph) is shown in Fig. \ref{Cycles}. The co-rank of this subgraph is $1$ (since the only component is a cycle, and the co-rank is the number of cycles) and the rank is even as the number of vertices in any Feigenbaum graph is always odd, and we only have one component). Thus this subgraph contributes $2$ in the sum (Eq. \ref{eq:det}) of the determinant of $F_{3}$. This is true for any $n$, i.e., the largest cycle contributes 2 towards the determinant.

We can construct other spanning elementary subgraphs by taking any other cycle (big or small as in Fig.~\ref{Cycles}) and joining the remaining bottom edges. As such cycles always contain an odd amount of vertices, we are always left with an even amount of vertices on the bottom, which permits us to join the rest of the vertices with single edges. Such spanning elementary subgraphs, for $F_3$ are shown in Fig.~\ref{fig:Six-other-configurations}.

\begin{figure}
\begin{centering}
\includegraphics[scale=0.5 ]{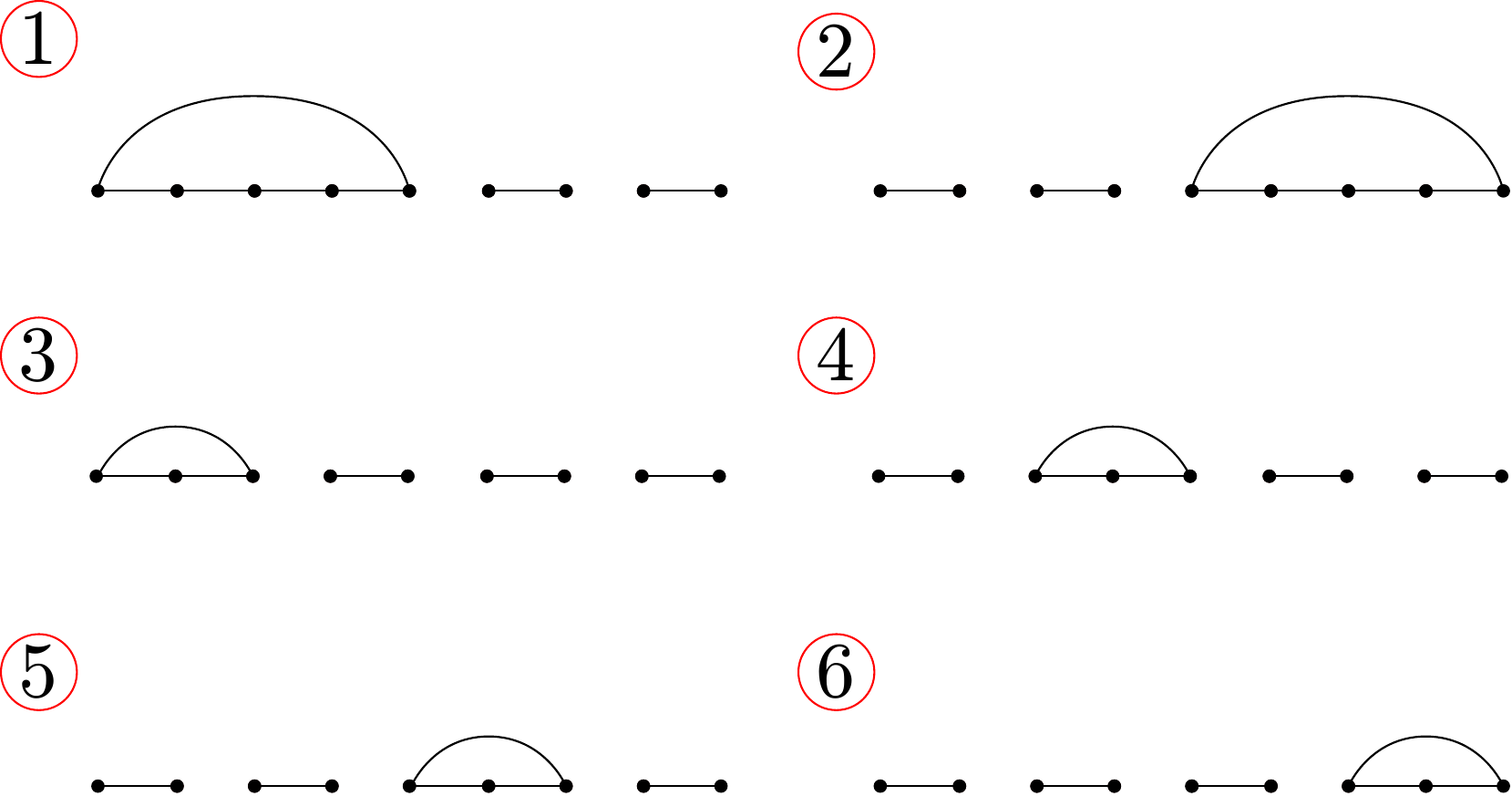}
\par\end{centering}
\protect\caption{\label{fig:Six-other-configurations}Six other configurations for
elementary spanning subgraphs of $F_{3}$ (other than the largest cycle).}
\end{figure}

We stipulate the following: \emph{we cannot have more than one cycle in any elementary subgraph}. This is because if we take two cycles in our elementary subgraph, we will be left with an odd amount of vertices. An odd amount of vertices cannot be joined only by single edges, thus we would require another cycle to give us an even amount of vertices. However, because of the construction of the Feigenbaum graphs, this will leave us with an odd amount on either side of one of the cycles, and this process repeats until we are left with a single node that cannot by introduced in to any spanning elementary subgraph.

By the same reasoning, any other cycles considered which are not listed in Fig.~\ref{Cycles} (for example in $F_3$, taking the triangle formed by the 1st, 3rd and 5th nodes) will again leave us with an odd amount of nodes, the first of which (by ordering the remaining nodes and numbering them left to right starting with 1) can only be connected by a single edge to the next node. This process repeats until we are left with a single node that cannot be introduced in to any spanning elementary subgraph.

Thus each spanning elementary subgraph contains only one of our big or small cycles, and each big or small loop corresponds to exactly one spanning elementary subgraph. The co-rank of all elementary subgraphs of all $F_{n}$ is therefore equal to 1. This concludes the proof of Lemma \ref{lema1}.
\end{proof}

\begin{lem}
\label{lem:Spanning-elementary-subgraphs}Spanning elementary subgraphs consisting of a big cycle have an even number of single edges components. S.E.S's consisting of a small cycle have an odd number of single edges components.
\end{lem}

\begin{proof}
The number of vertices of $F_{n}$ is $2^{n}+1$. The number of vertices in one of the big or small cycles is $2^{k}+1$ where $1\leq k<n$. The number of vertices remaining whe we add a cycle component (a loop) to a S.E.S is $2^{n}+1-(2^{k}+1)=2^{n}-2^{k}$ and the number of single edge components is 
\[
\frac{2^{n}-2^{k}}{2}=2^{n-1}-2^{k-1}
\]

which is even if and only if $k\neq1,$ i.e., only for big cycles. This concludes the proof of Lemma \ref{lem:Spanning-elementary-subgraphs}.
\end{proof}

Going back to our formula in Eq. \ref{eq:det}, we have $c(H)=1$, therefore

\[
det(A)=\sum_{H}2(-1)^{r(H)}
\]

In Lemma. \ref{lem:Spanning-elementary-subgraphs} we proved that $r(H)$ is even for E.S.Gs containing big cycles and odd for small. Using simple combinatoric arguments shown in Fig. \ref{fig:Diagram-showing-each}, we have for $n\geq2,$

\[
det(F_{n})=-2^{n}+\sum_{k=1}^{n-1}2^{k}=-2
\]

For $k>1$ the big and small cycles give slightly different contributions to the determinant, similar to Fig. \ref{fig:Diagram-showing-each}, however the number of big and small cycle is different. It can easily be checked that
\[
det(F_{n}^{k})=-2k,
\]
which concludes the proof of Theorem \ref{thm_ryan}.
\end{proof} 
\begin{figure}
\begin{centering}
\includegraphics[scale=0.5]{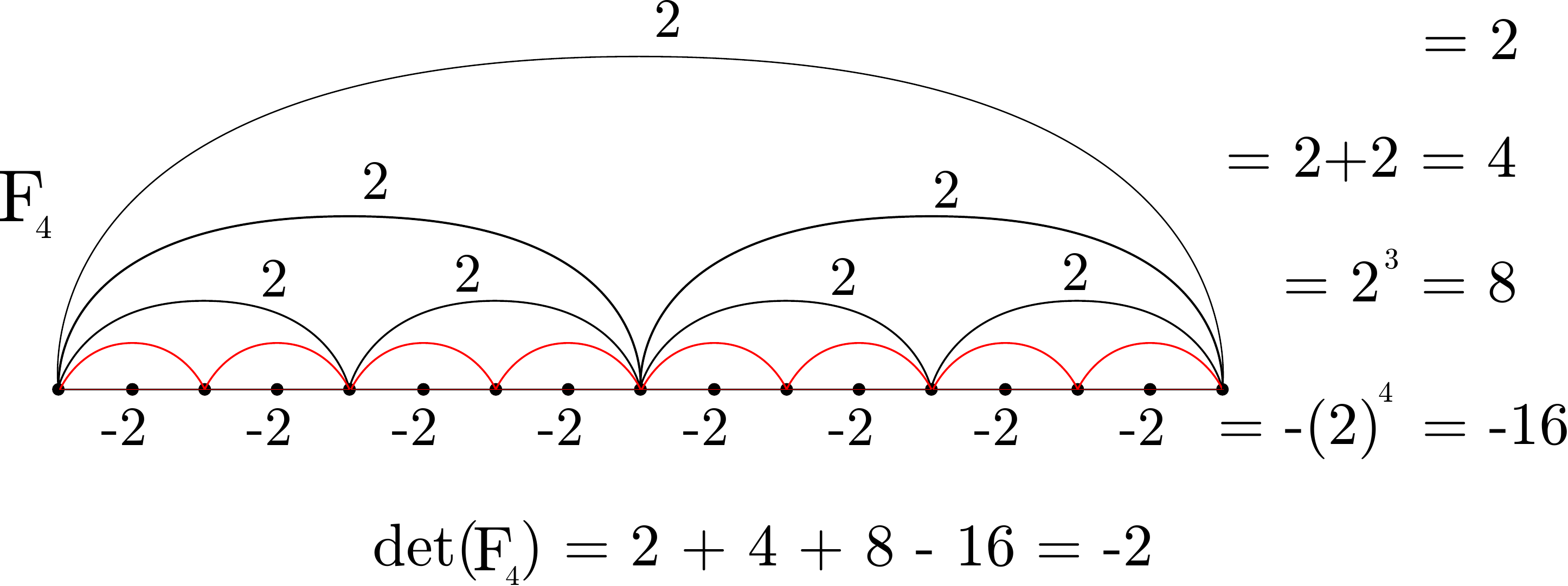}
\par\end{centering}
\protect\caption{\label{fig:Diagram-showing-each}Diagram showing each cycle's contribution
to $det(F_{n})$.}
\end{figure}

\begin{remark}
It can also be checked that, using similar arguments to the proof of Thm.~\ref{thm_ryan}, $det(F_1^k)=(-1)^{k+1}2k$ and that for $n=0$:
\[
det(F_{0}^k)=\begin{cases}
0 & k\text{ even}\\
(-1)^{\frac{k+1}{2}} & k\text{ odd}\\
\end{cases}
\]
\end{remark}

\section{$\mu>\mu_{\infty}$: Spectral properties of chaotic Feigenbaum graph ensembles}
\label{sec:chaos}

In this section we explore Feigenbaum graphs in the region $\mu>\mu_{\infty}$. As discussed in \cref{sec:intro_chaos}, for a given $\mu<\mu_{\infty}$, and for a given series size $N$, the resulting Feigenbaum graph was unique because the order in which the trajectory visits the stable branches of the periodic attractor is unique (indeed, it is universal for all unimodal maps, not just the logistic map, so $F_n^k$ are indeed universal \cite{Feig}). However, for $\mu>\mu_{\infty}$ this is no longer the case:  for a specific $\mu$, each initial condition will generate a priori a different chaotic trajectory, and hence a different Feigenbaum graph. Since in this case $n$ and $k$ do not apply anymore, we use the notation $F(\mu,N)$ to describe the ensemble of Feigenbaum graphs associated to a trajectory of size $N$ (so the corresponding HVG has $N$ vertices) generated by the logistic map with parameter $\mu$.

\subsection{Self-averaging properties of $\lambda_{\text{max}}$}
\label{sec:chaos1}

We start by exploring the self-averaging properties of the ensembles of Feigenbaum graphs. First, we fix $\mu=4$ (fully developed chaos) and extract an ensemble of 100 time series of series with $N= 2^{q}$ for $q=10,\dots,15$, each generated with a different initial condition. For each series, we then extract its Feigenbaum graph and calculate $\lambda_{\text{max}}$. For each time series size $N$, we compute the mean and standard deviation of the ensemble of $\lambda_{\text{max}}$. To assess whether this quantity self-averages as $N$ increases \cite{SA}, in the left panel of Figure \ref{fig:MeanStd} we plot the relative variance $R_\lambda$ as a function of $N$, defined as
$$R_\lambda(N)=\frac{\langle\lambda_{\text{max}}^2\rangle - \langle\lambda_{\text{max}}\rangle^2}{\langle\lambda_{\text{max}}\rangle^2},$$
where the average $\langle\cdot\rangle$ is performed over the ensemble of realisations.

 We observe that this quantity decreases with $N$, certifying that, for $\mu=4$, the largest eigenvalue is a self-averaging quantity. This means that with regards the largest eigenvalue, a {\it typical} realisation of $F(\mu=4, N)$ provides a faithful representation of the ensemble. Moreover, the relative variance scales as a power law $R_\lambda(N)= cN^{z}$ with $z\approx -0.221$, hence the system is {\it weakly} self-averaging (because we have $-1<z<0$).
 
A similar analysis is performed now for the whole range of values of $\mu>\mu_{\infty}$ for which the Lyapunov exponent ($\textsc{le}$) is positive (i.e., we discard periodic windows). In each case, a power law fit $R_\lambda(N)=cN^{z}$ is computed. In the right panel of Figure \ref{fig:MeanStd} we plot the estimated exponent $z(\mu)$. In most of the cases we find that the system remains weakly self-averaging. There is only one exception for this otherwise general behaviour: for a specific value of $\mu$ only slightly above $\mu_{\infty}$ ($\textsc{le}\approx 0+$) we find that $z>0$, i.e., the relative variance increases with $N$. This anomalous behaviour can be explained as follows: in the onset of chaos $\mu=\mu_{\infty}$, the Feigenbaum graph ensemble is still degenerate (i.e. only one unique configuration). As we enter into the chaotic region but remain very close to $\mu_{\infty}$, a  trajectory of the map will visit what is known as a ghost of the attractor found in the accumulation point. In fact, the structure of a realisation of a Feigenbaum graph just above the accumulation point is very similar to the one found at the accumulation point with just a few additional `chaotic' edges \cite{Feig}. The existence of these edges is what allows the ensemble in this case to no longer be degenerate. Now, the number of these chaotic edges will proportionally increase when the series size $N$ increases, simply because as $N$ increases the trajectory will show additional deviations from the ghost attractor. Accordingly, the total number of possible configurations of the ensemble of Feigenbaum graphs very close to the accumulation point increases from essentially one (degenerate case) when $N$ is small to many as $N$ increases. As a byproduct, the relative variance will necessarily increase as a function of $N$ in this case, hence $z>0$.

\begin{figure}
\begin{centering}
\includegraphics[scale=0.4]{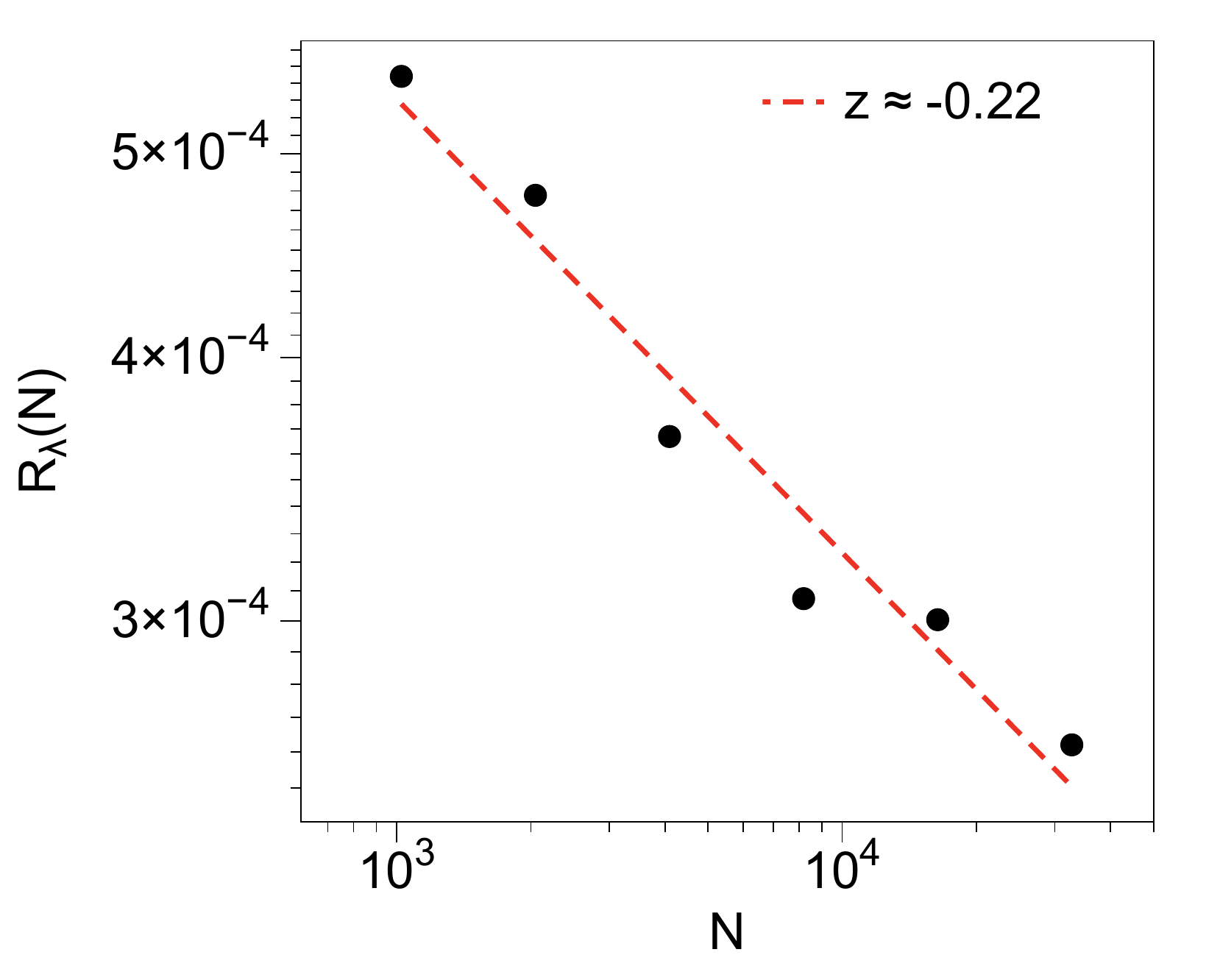}
\includegraphics[scale=0.4]{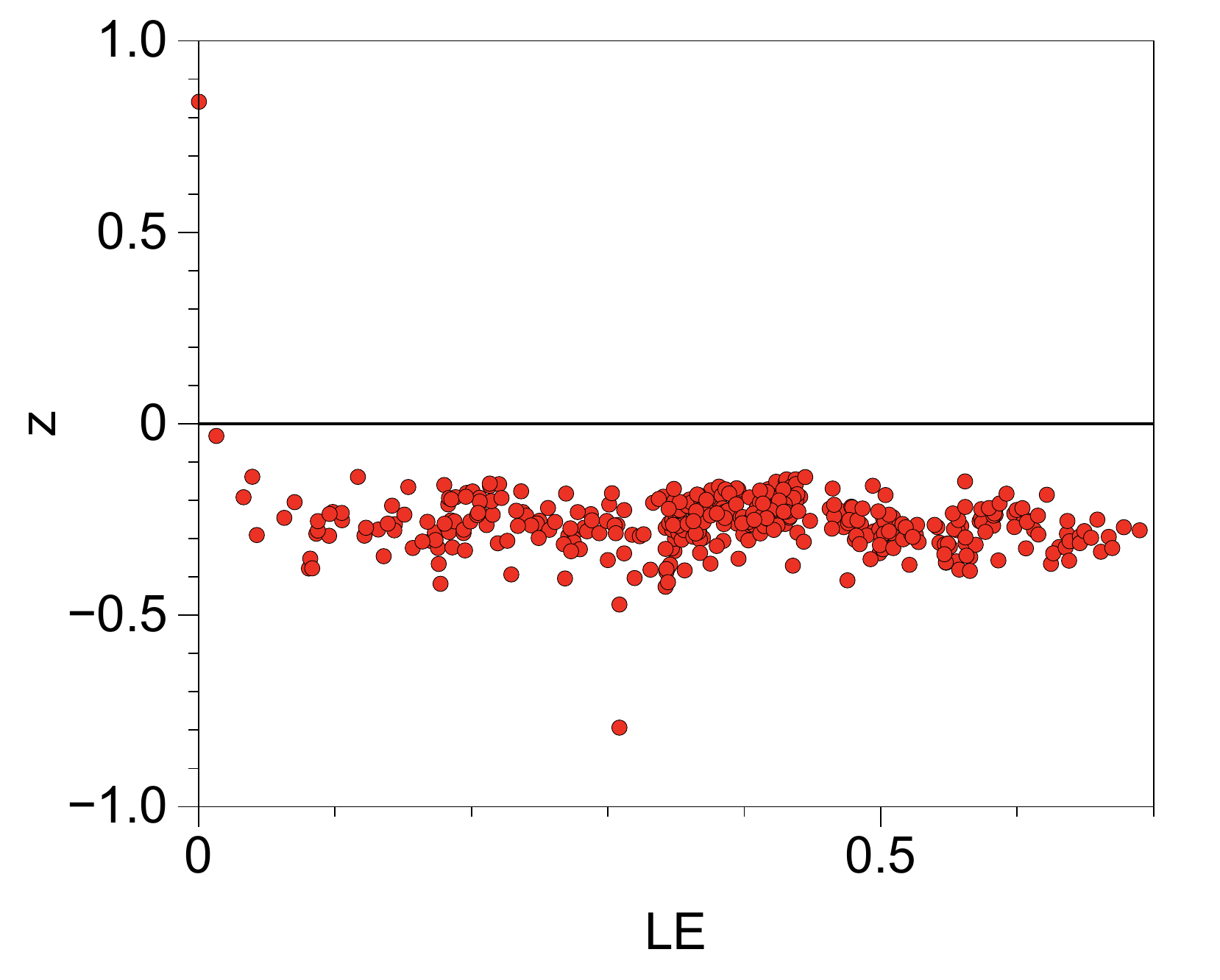}
\par\end{centering}
\protect\caption{\label{fig:MeanStd}(Left) Log-log Plot showing the relative variance $R_\lambda(N)$ as a function of size $N$, computed over an ensemble of 500 realisations of Feigenbaum graphs $F(4,2000)$. The curve is fitted by a power law $R_\lambda(N)= cN^{z}$, where the best fit provides $z\approx -0.22$, suggesting the system is weakly self-averaging ($-1<z<0$). (Right panel) Fitted exponent $z$ for the range of values of $\mu$ for which $\textsc{le}(\mu)>0$. In most of the cases we find $-1<z<0$, confirming that the system is weakly self-averaging. For $\textsc{le}(\mu)\approx 0+$ (which holds for $\mu$ only slightly above $\mu_{\infty}$, $z>0$. This can be explained in terms of the ghost structure present in the graphs (see the text).}
\end{figure}

\subsection{Searching spectral correlates of chaoticity}
\label{sec:chaos2}
One of the main motivations that leads us to explore the largest eigenvalue of HVGs is that some research claims that this is an informative quantity for the `complexity' of the associated time series (see for instance \cite{Fioriti, Ahmadlou2010, epilepsy, GIC2, GIC3}). If this was the case, we wonder if such quantity is able to quantify the `degree of chaoticity' of a given (chaotic) time series. Within the realm of nonlinear time series analysis, a relevant property that quantifies how chaotic a system is the \textit{sensitivity to initial conditions}, better described by the largest Lyapunov exponent of the system which accounts for the (exponential) separation rate of two initially nearby trajectories. For univariate time series extracted from a map $x_{i+1}=f(x_i)$, there is only one Lyapunov exponent \textsc{le}, which can be estimated from a single (long) time series as \cite{strogatz}
$$\textsc{le}={\displaystyle\lim _{N\to \infty }{\frac {1}{N}}\sum _{i=0}^{N-1}\log |f'(x_{i})|}$$
Thus, for each $\mu \in [\mu_{\infty},4]$ (sampled in steps of $\Delta \mu=0.001$) we have generated a single trajectory, and computed both the \textsc{le} and $\lambda_{\text{max}}$. In figure \ref{fig:ScatterLELambda} we show the scatter plot of $\lambda_{\text{max}}$ vs \textsc{le}. Surprisingly, no obvious correlation emerges in this picture, which suggests that $\lambda_{\text{max}}$ does not correlate to the sensitivity to initial conditions.\\

\begin{figure}
\begin{centering}
\includegraphics[scale=0.5]{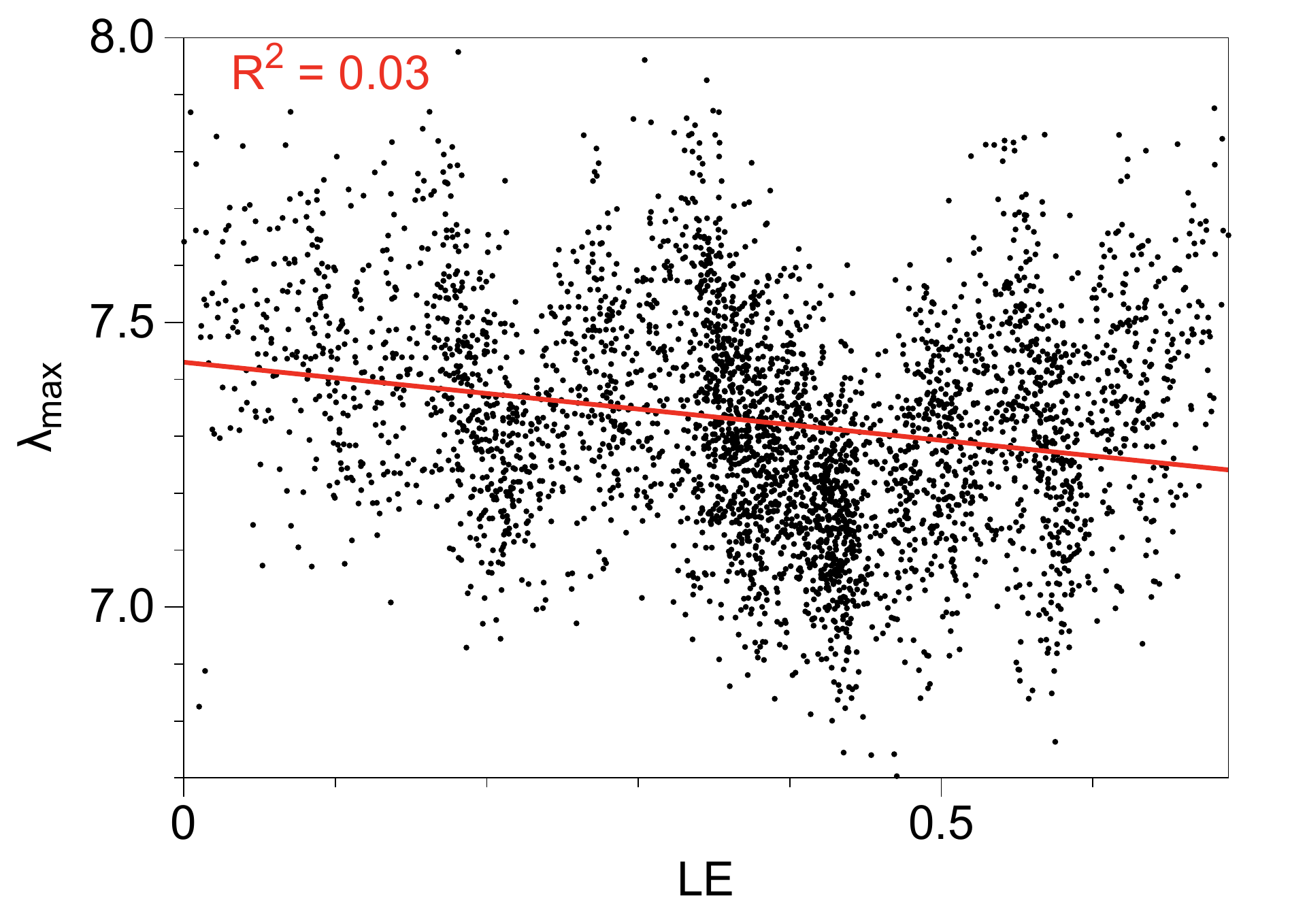}
\par\end{centering}
\protect\caption{\label{fig:ScatterLELambda}Scatter graph of the maximum eigenvalue $\lambda_{\text{max}}$ of $F(\mu,2000)$ vs the Lyapunov exponent $\textsc{le}(\mu)$, for values of $\mu\in [\mu_{\infty}$,4] in steps of $\Delta \mu=0.001$ (only positive values are selected to avoid periodic windows). No correlation emerges.}
\end{figure}

Does this mean that HVGs are not inheriting chaoticity properties, or that these are simply not inherited in $\lambda_{\text{max}}$? As a matter of fact, previous works have shown that the HVGs \textit{do} capture chaoticity, as $\textsc{le}(\mu)$ is very well approached (from above) by suitable block-entropies of the Feigenbaum graph's degree sequence \cite{wolfram}. So the question is whether the spectral properties of these graphs are able to capture such properties. We do not have a definite answer for this, but let us comment that we have checked scatter plots similar to Figure \ref{fig:ScatterLELambda} for other spectral properties, such as the graph's Von Neumann entropy \cite{severini_VN}, spectral gap or the (logarithmic) tree number, with similarly unsuccessful results (data not shown). Hence our partial conclusion is that spectral properties do not quantify different levels of chaoticity. The natural question is therefore: do these characterise chaos {\it at all}? To address this question, in the next and final section of the paper we will make a systematic comparison between the spectral properties of Feigenbaum graphs associated to chaotic series and those of generic HVGs associated to random uncorrelated series.

%\subsection{$\lambda_{max}$ and Von Neumann entropy are not related to standard dynamical properties}
%We now calculate $\lambda_{\text{max}}$ for the HVG of the logistic map for values of $\mu \in [\mu_{\infty},4]$ in steps of $0.001$. We also calculate the Lyapunov exponent of each logistic map. We plot a scatter graph $(x,y)=(\lambda_{\text{max}}(\mu),LE(\mu))$, for values of $\mu$ where $LE(\mu)$ is positive. This is shown in Fig. \ref{fig:ScatterLELambda}. We observe no correlation.

%We now calculate the Von Neumann entropy (VN) for the HVG of the logistic map for values of $\mu \in [\mu_{\infty},4]$ in steps of $0.001$. We compare this to the Lyapunov exponent of the logistic map for each HVG. We plot a scatter graph $(x,y)=(VN(\mu),LE(\mu))$, for values of $\mu$ where $LE(\mu)$ is positive. This plot is shown in Fig. \ref{fig:VonNeumann}.

%\begin{figure}
%\begin{centering}
%\includegraphics[scale=0.6]{LogisticGraphs/VonNeumann.pdf}
%\par\end{centering}
%\protect\caption{\label{fig:VonNeumann}Scatter graph of the Von Neumann entropy of the HVG for the logistic map against the Lyapunov exponent for values of $\mu$ ranging between $\mu_{\infty}$ and $4$ in steps of $0.001$}
%\end{figure}

\subsection{Comparison with iid}
In \cref{sec:chaos2} we came to the conclusion that spectral properties don't seem to characterise (in a quantitative way) the chaoticity of the series. Hence the question: do they carry {\it qualitative} information, or on the contrary, spectral properties do not distinguish between chaotic series and random ones? If this was to be the case, the spectral properties shouldn't differ much from what we would find for random, uncorrelated series (iid).

\subsubsection{$\lambda_{max}$}
First let us note that  in \cite{Fioriti} the authors explored whether $\lambda_{\text{max}}$ could distinguish chaotic and random series, with interesting numerical evidence suggesting that indeed chaos can be distinguished from an iid process under this lens. As a cautionary note, observe however that their analysis was based on estimating $d_{\text{max}}$, as they claim that $\lambda_{\text{max}}\approx \sqrt{d_{\text{max}}}$ when $N\to \infty$. This is however not true in general (for a generic graph), and in the context of HVGs it is actually unknown. Also, they assumed that this quantity converged as the series size $N$ increases, and numerically checked this in a small interval of $N$. Note, however, that analytical results \cite{nonlinearity} suggest that $d_{\text{max}}$ is unbounded for both iid and chaotic processes as their degree distribution has an exponential tail (that is to say, in order to find a certain value for  $d_{\text{max}}$ one just needs to increase (exponentially) the series size $N$).

\begin{figure}[h]
\begin{centering}
\includegraphics[scale=0.5]{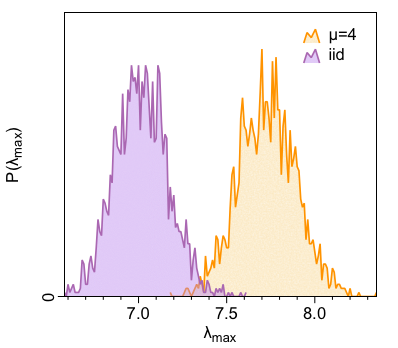}
\par\end{centering}
\protect\caption{\label{fig:ensemble} Ensemble histogram $P(\lambda_{\text{max}})$ for $\mu=4$ and iid (the histogram has been smoothed).}
\end{figure}

Does $\lambda_{\text{max}}$ converge as $N\to\infty$? Since both iid and chaotic series are aperiodic, from Eq.(\ref{periodic_k}) we get $\bar d=4$ in both cases. Furthermore, from \cite{PRE, nonlinearity} it is known that degree distribution for both infinite iid and a chaotic process such as the logistic map has an exponential tail, with $P(d)\sim \exp(-\gamma d)$. In particular, for $\mu=4$ a good approximation is $\gamma_{\mu=4}\approx \log(4/3)$, whereas for an iid process the exponential distribution is exact and $\gamma_{iid}=\log(3/2)$. Note, however that these expressions hold in the limit $N\to \infty$, where $d_{\text{max}}$ is unbounded (although it grows rather slowly with $N$) in both cases, suggesting that $\lambda_{\text{max}}$ is indeed unbounded in the limit $N\to \infty$. This is not unexpected, as $F(\mu,\infty)$ are not locally finite. For that reason, in order to assess whether $\lambda_{\text{max}}$ can indeed distinguish chaos from iid, we shall analyse finite trajectories ($N<\infty$). $d_{\text{max}}$ is therefore the largest possible degree of $F(\mu,N)$. Statistically speaking, we can state that $d_{\text{max}}$ is only reached once in the whole graph, and therefore $d_{\text{max}}$ should fulfil
$$P(d=d_{\text{max}})\cdot N = 1$$
A quick calculation yields
\begin{equation}
d_{\text{max}} \sim \frac{\log N}{\gamma}
\end{equation}
and according to Eqs.(\ref{bound_degree_eq1}) and $(\ref{bound_degree_eq2})$, we have for both iid and chaos:
\begin{equation}
\sqrt{d_{\text{max}}} \leq \lambda_{\text{max}}\leq \sqrt{d_{\text{max}}(d_{\text{max}}-1)}
\end{equation}
Interestingly, the difference between the chaotic case $(\mu>\mu_{\infty})$ and the random case (iid) is evident in $d_{\text{max}}$:
\begin{equation}
d_{\text{max}}^{\mu} = \frac{\gamma_\mu}{\gamma_{\text{iid}}} d_{\text{max}}^{\text{iid}} 
\end{equation}
which for $\mu=4$ becomes
$$d_{\text{max}}^{\mu=4}\approx 1.4\cdot d_{\text{max}}{\text{iid}} $$

We fix $N=2000$ and compute $\lambda_{\text{max}}$ for iid and $\mu=4$ over 2000 realisations. We plot the resulting histograms are in Figure \ref{fig:ensemble}, finding $\langle \lambda_{\max}^{\text{iid}} \rangle=7.01\pm 0.15$,  and $\langle \lambda_{\max}^{\mu=4} \rangle=7.73\pm 0.16$. The two quantities are clearly different.\\

We now assess whether $\lambda_{\max}$ of an ensemble of logistic maps is systematically different than the same quantity obtained from iid. To do this, we consider all values of $\mu$ for which $\textsc{le}(\mu)>0$ and for each of these values, we have performed a 2-sampled t-test between $\langle \lambda_{\max}^{\mu} \rangle$ and $\langle \lambda_{\max}^{\text{iid}} \rangle$, and obtained a p-value for each test. We systematically find very small p-values, concluding that $\langle \lambda_{\max}\rangle$ can indeed distinguish time series extracted from the whole chaotic region from a purely random process.

\begin{figure}[h]
\begin{centering}
\includegraphics[scale=0.4]{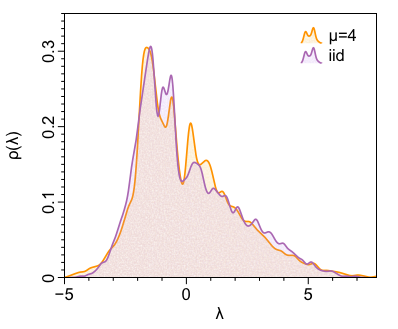}
\includegraphics[scale=0.4]{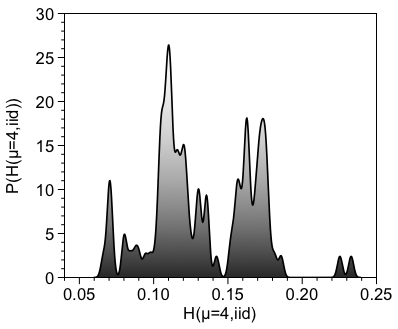}
\par\end{centering}
\protect\caption{\label{fig:HistogramEigenvalues}(Left Panel) Histogram showing the distribution of eigenvalues $\rho(\lambda)$ for the Feigenbaum associated to the logistic map ($\mu=4$) in orange (time series of $N=2000$), versus the one associated to an iid time series of the same size in purple. To help the eye distinguish both distributions, a smoothing has been applied. (Right Panel) Ensemble distribution of the Hellinger distance $H(\mu=4, \text{iid})$ between $\rho(\lambda)$ for $\mu=4$ and an iid process, for a total of 100 realisations. The mean of the ensemble is $\langle H(\mu=4,\text{iid}) \rangle=0.13\pm 0.03$ }
\end{figure}

\subsubsection{Distribution of eigenvalues}
To round off our analysis, we now compare the distribution of eigenvalues in the chaotic case to the one obtained for random iid time series of the same size. We start with $\mu=4$. 
We extract a time series of size $N=2000$ for each process, compute the list of eigenvalues and display their frequency $\rho(\lambda)$ in a histogram. These are shown in the left panel of Fig. \ref{fig:HistogramEigenvalues}.  We observe that the distribution is somewhat different for specific ranges. To quantify `how different' they are, we compute the Hellinger distance, defined as
$$H(p,q)=\sqrt{1-\sum_x \sqrt{p(x)\cdot q(x)}},$$
where $p(x)$ and $q(x)$ are two sample distributions. After an ensemble average over 100 realisations, the average Hellinger distance between $\mu=4$ and iid is $H(\mu=4,\text{iid})=0.13\pm 0.03$ (see the right panel of Fig. \ref{fig:HistogramEigenvalues} for the ensemble distribution of Hellinger distances).\\

Finally, we explore the distance for $\mu \in [\mu_{\infty},4]$. A scatter plot of $H(\mu, \text{iid})$ vs $\textsc{le}(\mu)$, for those values for which the Lyapunov exponent is positive is shown in figure \ref{fig:last}. Unexpectedly, a clear negative correlation emerges between $H(\mu, \text{iid})$ and $\textsc{le}(\mu)$. The best linear fit is $H(\mu, \text{iid}) \approx 0.21627 - 0.23208\textsc{le}(\mu)$. While a sound theoretical justification for this negative correlation is left for future work, heuristically one can say that the larger the Lyapunov exponent, the more chaotic the time series is and thus the less easy is to distinguish the spectrum of the associated Feigenbaum graph from the one generated from a random series.

\begin{figure}[h]
\begin{centering}
\includegraphics[scale=0.5]{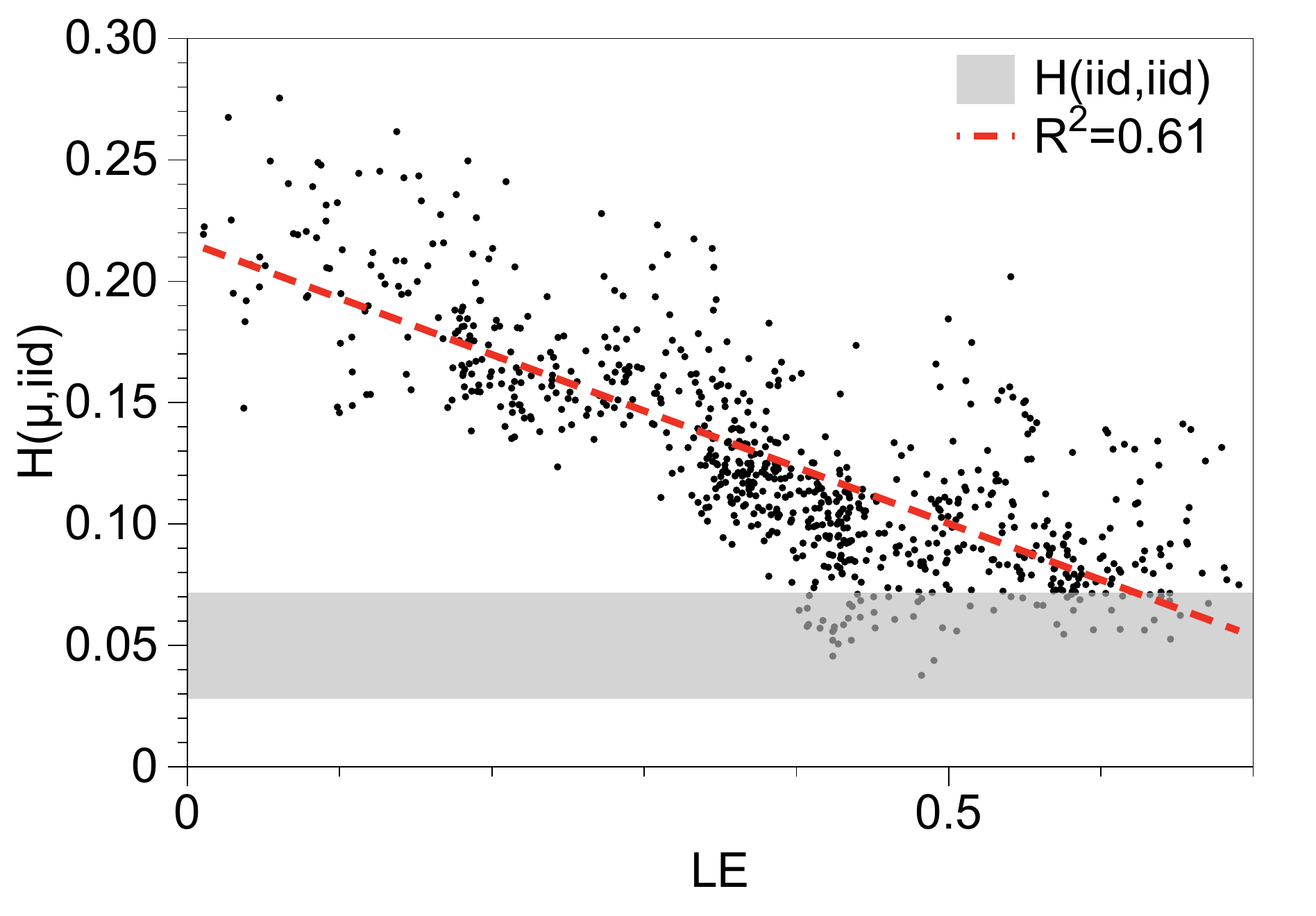}
\par\end{centering}
\protect\caption{\label{fig:last}Scatter plot of the Hellinger distance $H(\mu,\text{iid})$ between the eigenvalue distribution of the Feigenbaum graph $F(\mu, 2000)$ and the one associated to a random iid time series of the same size as a function of the Lyapunov exponent $\textsc{le}(\mu)$, for those values $\mu \in [\mu_{\infty},4]$ for which the Lyapunov exponent is positive (sampling $\Delta \mu=5\cdot10^{-4}$). For comparison, $H(\text{iid},\text{iid})$ is also shown (the gray area denotes $\langle H(\text{iid},\text{iid}) \rangle \pm std=0.05\pm 0.02$). This area denotes the range for which distributions cannot be distinguished.
A clear negative correlation between the Hellinger distance to iid and the Lyapunov exponent emerges.}
\end{figure}

\section{Discussion}
\label{sec:conclusion}
Horizontal Visibility Graphs (HVGs) have been widely used as a method to map a time series into a graph representation, with the aim of performing graph-based time series analysis and time series classification. Among other properties, the Graph Index Complexity --(GIC), a rescaled version of the maximal eigenvalue of the HVG's adjacency matrix-- has been used as a network quantifier in several applications. However, there is a shortage of theoretical analysis of the spectral properties of HVGs, as most works essentially deal with applications of GIC for real-world time series classification.\\
Here we make the first step to partially fill this gap by addressing the spectral properties of HVGs associated to certain classes of periodic and chaotic time series. For convenience, we focus on the archetypal logistic map as it is a canonical system producing periodic time series of different periods and chaotic time series with different degrees of chaoticity (i.e, different Lyapunov exponent) as it undergoes the Feigenbaum scenario.\\ We were able to enumerate the visibility graphs below the map's accumulation point in terms of a bi-parametric family of finite Feigenbaum graphs $F_n^k$, and have explored their spectral properties (in particular, the behaviour of the maximal eigenvalue of the adjacency matrix) as a function of $n$ and $k$. We found noteworthy patterns, and numerical results were complemented with analytical developments as well as exact results. Other aspects that were investigated include the full spectrum, the determinant, the number of distinct eigenvalues, and the number of spanning trees of the whole family of $F_n^k$.\\
A similar analysis was then conducted in the region of the map's parameter where trajectories are chaotic, finding that the maximal eigenvalue, while being a good discriminator between chaos and noise, is not able to quantify chaoticity. The eigenvalue distribution, on the other hand, was found to carry information about time series chaoticity, in particular its Lyapunov exponent.\\
In this work we have also outlined a number of conjectures and open problems which we hope will trigger some attention in the algebraic and spectral graph theory community.

\begin{appendix}

\section{Walks of $F_n$}\label{app:walks}

\subsection{Maximum 2-walks}
In Section.~\ref{sec:gelfand} we use the result $\Vert({\bf A}_n)^2\Vert_{\infty}=2n^2+2$, and we prove it here. Note that $\Vert({\bf A}_n)^2\Vert$ is the maximum of the number of 2-walks originating at a node, over all the nodes. It is clear that this node is the central node, which we will call $v_c$, which has degree $2n$ (as we prove in Prop.~\ref{prop:basic}). Also note that to count the number of 2-walks originating from $v_c$, we can count the total degree of the neighbours of $v_c$. We can observe that apart from the boundary (left and right) nodes, which have degree $n+1$, the degrees of the neighbours of the $v_c$ are $2, 4, 6,\dots,n-1$ (and these are counted twice). Summing up all these degrees we have
\begin{align*}
\Vert({\bf A}_n)^2\Vert &= 2(n+1)+\sum_{k=1}^{n-1}2k\\
				      &= 2(n+1)+2 \cdot 2 \cdot \frac{n(n-1)}{2}\\
					  &= 2(n+1)+2n(n-1)\\
					  &= 2n^2 + 2
\end{align*}

\subsection{Coefficients}
In section.~\ref{sec:walkbound} we defined $a(n)$, $b(n)$, $c(n)$ and $d(n)$ to be the total number of 3-walks, 2-walks, 1-walks and 0-walks respectively. We state that 
\begin{eqnarray}
&&b(n) = 24\cdot 2^n - 2n^2 - 12n - 22 \nonumber \\    %number of 2 walks for F_n
&&c(n) = 4\cdot 2^n - 2  \nonumber \\ %number of 1-walks for F_n (= twice number of edges)
&&d(n) = 2^n + 1. \nonumber  %number of 0-walks for F_n (= number of nodes).
\end{eqnarray}
Observe that the number of 0-walks, $d(n)$, is the number of nodes, which is equal to $2^n+1$. The number of 1-walks, $c(n)$, is twice the number of edges, and is equal to $2(2^{n+1}-1) = 4 \cdot 2^n-2$.

Reaching the formula for $b(n)$ is a little trickier. First define $I_n^1$ to be the $n \times n$ matrix with zeros everywhere, except a 1 in the top right entry. Similarly define ${}_{1}I_n$ to to be the $n \times n$ matrix with zeros everywhere, except a 1 in the bottom left entry. Notice that
\begin{equation}\label{walkeq}
{\bf A}_n={\bf A}_{n-1}^2 + I_n^1 + {}_{1}I_n
\end{equation}
where ${\bf A}_{n-1}^2$ is the adjacency matrix of $F_{n-1}^2= F_{n-1} \oplus F_{n-1}$ as in Definition~\ref{concatenation}. The quantity we wish to find is $\sum_{i,j}({\bf A}_{n-1}^2)^2$. We plot a visualisation of the matrix $({\bf A}_{n-1}^2)^2$ in Figure.~\ref{fig:walkproof}. The source of the top left and bottom right blocks $({\bf A}_{n-1})^2$ should be clear by studying the form of the matrix ${\bf A}_{n-1}^2$. A matrix $\bf C$ appears in the top right (with its transpose in the bottom left), and has size $2^{n-1} \times 2^{n-1}$. The origin of this matrix is not immediately clear, however it is the matrix obtained when the central row vector of ${\bf A}_{n-1}^2$ hits itself under squaring. The vector has sum $2n$ (recall the degree of the central vertex), but only half of the vector hits itself when creating $\bf C$, so we will only consider the first $2^{n-1}$ values of this central vector, and we will call it $v_c$. Thus $\bf C$ is the matrix where the $n$th row vector is $v_c$ if the $n$th value of $v_c$ is 1, and is zero otherwise (or rather, a vector of zeros of length $2^{n-1}$), or equivalently 
\[
{\bf C}=\underbrace{(v_c^\intercal|v_c^\intercal|\cdots|v_c^\intercal)}_{2^{n-1}}.
\]
The sum of $v_c$ is $n$, so we have that $\sum_{i,j}{\bf C}=\sum_{i,j}{\bf C^\intercal}=n^2$. Going back to Equation.~\ref{walkeq} we have
\begin{align*}
({\bf A}_n)^2 &= ({\bf A}_{n-1}^2 + I_n^1 + {}_{1}I_n)^2\\
		&= ({\bf A}_{n-1}^2)^2 + {\bf A}_{n-1}^2 \cdot I_n^1 + {\bf A}_{n-1}^2 \cdot {}_{1}I_n + I_n^1 \cdot {\bf A}_{n-1}^2\\
		& + (I_n^1)^2 + I_n^1 \cdot {}_{1}I_n + {}_{1}I_n \cdot {\bf A}_{n-1}^2 + {}_{1}I_n \cdot I_n^1 + ({}_{1}I_n)^2
\end{align*}
Summing this quantity over $i$ and $j$, the first term gives a contribution of twice that of ${\bf A}_{n-1}^2$ (see Figure.~\ref{fig:walkproof}) and twice the sum of $\bf C$. The terms involving ${\bf A}_{n-1}^2$ and either $I_n^1$ or ${}_{1}I_n$ give us a contribution of $n$, as these vectors extract the top, bottom, left and right row/column vectors of ${\bf A}_{n-1}^2$ and these have sum $n$ (recall the degree of the left or right boundary nodes is $(n-1)+1=n$). The terms $(I_n^1)^2$ and $({}_{1}I_n)^2$ have sum $0$ but the terms $I_n^1 \cdot {}_{1}I_n$ and ${}_{1}I_n \cdot I_n^1$ have sum $1$ each. Putting this together we have
\[
\sum_{i,j}({\bf A}_n)^2=2\cdot\sum_{i,j}({\bf A}_{n-1})^2 + 2n^2 + 4n + 2
\]
and writing $\sum_{i,j}({\bf A}_n)^2=b(n)$ we have a recurrence relation
\[
b(n)=2 \cdot b(n-1)+2n^2+4n+2
\]
We have that $b(0)=2$, hence this can be solved and we get
\[
b(n) = 24\cdot 2^n - 2n^2 - 12n - 22
\]
completing the proof. Constructing a similar proof for the 3-walks $a(n)$ could be possible but we were not able to. However, we can guess that $a(n)$ is of the form $m\cdot2^n+g(n)$ where $g(n)$ is a polynomial in $n$ and $m$ is an integer. Calculating $a(n)$ directly for enough values of $n$ we can estimate the coefficients, and indeed we do find integer coefficients with
\[
a(n) = 160\cdot2^n - 4n^3 - 30n^2 - 104n - 158
\]
We can extend this analysis by trying to create formulas for the 0, 1, 2, and 3-walks of $F_n^k$ which we define as $d(n,k)$, $c(n,k)$, $b(n,k)$ and $a(n,k)$ respectively. We immediately get $d(n,k)$ and $c(n,k)$ from the definitions. We can estimate the other two by proceeding with the same method as for $a(n)$ by guessing the general form of $b(n,k)$ and $a(n,k)$ to be $m\cdot2^n+g(n)+ k \cdot h(n)$ where $m$ is an integer and $g(n)$ and $h(n)$ are polynomials. The new contribution of $h(n)$ comes from adding $k$ copies of a certain number of walks. This yields
\begin{eqnarray}
&&a(n,k) = (320\cdot2^n - 4n^3 - 48n^2 - 196n - 312) + (k - 2)(160\cdot2^n - 16n^2 - 88n - 152) \nonumber \\  %number of 3-walks for F_nk
&&b(n,k) = 24\cdot 2^n - 2n^2 - 12n - 22 + (k - 1)(24\cdot2^n - 8n - 20)\nonumber \\    %number of 2 walks for F_nk
&&c(n,k) = 4k\cdot2^n - 2k  \nonumber \\ %number of 1-walks for F_nk (= twice number of edges)
&&d(n,k) = k\cdot2^n + 1. \nonumber  %number of 0-walks for F_nk (= number of nodes)
\end{eqnarray}

\begin{figure}[h]
\begin{centering}
\includegraphics[scale=0.8]{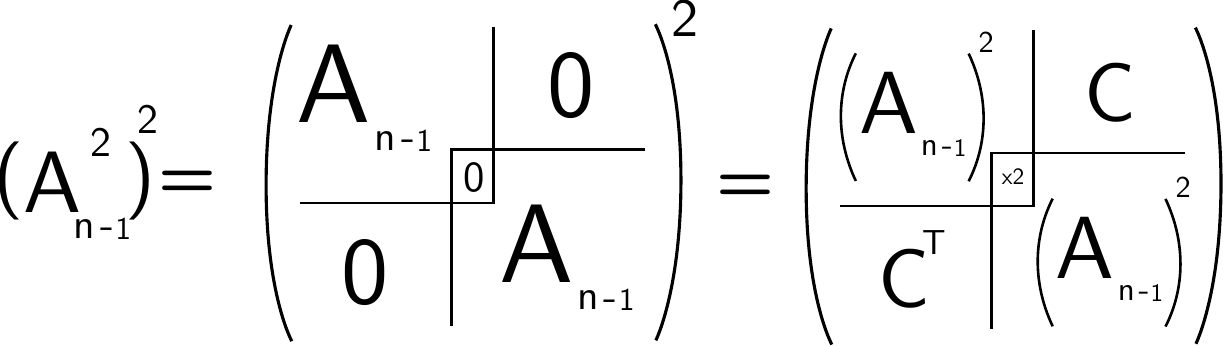}
\par\end{centering}
\protect\caption{\label{fig:walkproof}Diagram of the matrix representation of $(F_{n-1}^2)^2$ in terms of the matrices ${\bf A}_{n-1}$. The middle entry is the sum of the bottom right entry and top left entry of $({\bf A}_{n-1})^2$ (or equivalently twice the either entry as the matrix is symmetric). An extra matrix $C$ appears in the top right and bottom left blocks, whose entries sum to $n^2$, as explained in the text.}
\end{figure}

\end{appendix}

\noindent {\bf Acknowledgments. }{LL acknowledges funding from EPSRC Early Career Fellowship EP/P01660X/1. RF acknowledges doctoral funding from EPSRC.}

%\bibliographystyle{plain}
%\bibliography{FeigenbaumRefs}

\begin{thebibliography}{100}
\bibitem{PR_review} Yong Zou, Reik V. Donner, Norbert Marwan, Jonathan F. Donges,
Jürgen Kurths, Complex network approaches to nonlinear time series analysis, {\it Physics Reports} (2018)

\bibitem{PRE} B. Luque, L. Lacasa, F. Ballesteros, J. Luque, Horizontal visibility graphs: Exact results for random time series, {\it Physical Review E} { 80(4)} (2009) 046103.
\bibitem{PNAS} L. Lacasa, B. Luque, F. Ballesteros, J. Luque and JC Nu\~{n}o, From time series to complex networks: The visibility graph, \textit{Proc. Natl. Acad. Sci. USA} {105}, 13: 4972-4975 (2008).
\bibitem{Feig} B. Luque, L. Lacasa, F. Ballesteros, A. Robledo,  Analytical properties of horizontal visibility graphs in the Feigenbaum scenario, {\it Chaos} {\bf 22}, 1 (2012) 013109.

\bibitem{nonlinearity} L. Lacasa, On the degree distribution of horizontal visibility graphs associated to Markov processes and dynamical systems: diagrammatic and variational approaches, {\it Nonlinearity} {\bf 27} (2014) 2063-2093.
\bibitem{Luque2016} B. Luque, L. Lacasa, Canonical horizontal visibility graphs are uniquely determined by their degree sequence, {\it Eur. Phys. J. Sp. Top.} 226, 383 (2017).
\bibitem{wolfram} L. Lacasa and W. Just, Visibility graphs and symbolic dynamics, {\it Physica D} 374 (2018), pp. 35--44.

\bibitem{GIC} J. Kim and T. Wilhelm T, What is a complex graph? {\it Physica A} 387, 11 (2008).
\bibitem{Ahmadlou2010} Ahmadlou M, Adeli H, Adeli A. New diagnostic EEG markers of the Alzheimer's disease using visibility graph. {\it Journal of neural transmission} 117(9) (2010)1099-109.
\bibitem{epilepsy} X. Tang, L. Xia, Y. Liao, W. Liu, Y. Peng, T. Gao, and Y. Zeng, New Approach to Epileptic Diagnosis Using Visibility Graph of High-Frequency Signal, {\it Clinical EEG and Neuroscience}
44(2) (2013)150-156
\bibitem{Fioriti} Fioriti, V., Tofani, A. and Di Pietro, A., Discriminating chaotic time series with visibility graph eigenvalues. {\it Complex Systems} 21, 3 (2012).
\bibitem{GIC2} M. Mozaffarilegha, H. Adeli, Visibility graph analysis of speech evoked auditory brainstem response in persistent developmental stuttering. {\it Neuroscience Letters}, 696 (2019) 28-32.
\bibitem{GIC3} M. Nasrolahzadeh, Z. Mohammadpoory, and J. Haddadnia, J, Analysis of heart rate signals during meditation using visibility graph complexity. {\it Cognitive Neurodynamics}, 13, 1 (2019) 45-52.

\bibitem{spectral_infinite} B. Mohar and W. Woess, A survey on spectra of infinite graphs, {\it Bull. London Math. Soc.} 21 (1989) pp.209--234.
\bibitem{torgasev} A.Torgasev, On spectra of infinite graphs, {\it Publ. Inst. Math. Beograd} 29 (1981) pp. 269--282.

\bibitem{graph_wilson} R. J. Wilson, Introduction to graph theory, {\it Longman} (1972).

\bibitem{hwang2004cauchy} S. G. Hwang, Cauchy's interlace theorem for eigenvalues of Hermitian matrices, {\it The American Mathematical Monthly} 111(2) 2004 pp. 157--159.

\bibitem{das2004some} K. C. Das and Kumar. P, Some new bounds on the spectral radius of graphs {\it Discrete Mathematics} 281.1-3 (2004) pp.149--161.


%\bibitem{Bianchi} F.M. Bianchi, L. Livi, C. Alippi, and R. Jenssen, Multiplex visibility graphs to investigate recurrent neural network dynamics, {\it Scientific reports} 7, 44037 (2017).
%\bibitem{multi} L. Lacasa, V. Nicosia and V. Latora, Network Structure of Multivariate Time Series, {\it Scientific Reports} 5, 15508 (2015)


%\bibitem{motifs} J. Jacovacci, L. Lacasa, Sequential visibility graph motifs, {\it Physical Review E} 93, 042309 (2016)
%\bibitem{motifs2} J. Jacovacci, L. Lacasa, Sequential motif profile of natural visibility graphs {\it Physical Review E} 94, 052309 (2016)
%\bibitem{Newmanbook} M. Newman, The structure and function of complex networks,  {\it SIAM Review} 45, 167-256 (2003).
%\bibitem{Kurths2017} Z-K Gao, M. Small and J. Kurths, Complex network analysis of time series, EPL 116, 5 (2017).
\bibitem{severini} S. Severini, G. Gutin, T. Mansour, A characterization of horizontal visibility graphs and combinatorics on words, {\it Physica A} {\bf 390}, 12  (2011) 2421-2428.
%\bibitem{epl}  L. Lacasa, B. Luque, J. Luque and J.C. Nu\~{n}o, The Visibility Graph: a new method for estimating the Hurst exponent of fractional Brownian motion,
%{\it EPL} 86 (2009) 30001.
%\bibitem{quasi} B. Luque, A. N\'{u}\~{n}ez, F. Ballesteros, A. Robledo, Quasiperiodic Graphs: Structural Design, Scaling and Entropic Properties, {\it Journal of Nonlinear Science} {\bf 23}, 2, (2012) 335-342.
%\bibitem{pre2013} A.M. N\'{u}\~{n}ez, B. Luque, L. Lacasa, J.P. G\'{o}mez, A. Robledo, Horizontal Visibility graphs generated by type-I intermittency, {\it Phys. Rev. E}, {\bf 87} (2013) 052801.

\bibitem{EnzoBook}  Piet van Mieghem, {\it Spectra for complex networks} pp. 49 (Cambridge University Press, 2010).

\bibitem{biggs1993algebraic}  N. Biggs, {\it Algebraic graph theory} (Cambridge university press, 1993).

\bibitem{SA} A. Aharony and A. Brooks Harris, Absence of Self-Averaging and Universal Fluctuations in Random Systems near Critical Points, {\it Phys. Rev. Lett.} 77, 18 (1996).

\bibitem{strogatz} S.H. Strogatz, {\it Nonlinear dynamics and chaos} (Perseus books, Massachussets, 1994).

\bibitem{severini_VN} S. L. Braustein, S. Gosh and S. Severini, The Laplacian of a Graph as a Density Matrix: A Basic Combinatorial Approach to Separability of Mixed States, {\it Ann. Comb.} 10, 3 (2006) pp. 291--317.

\end{thebibliography}

\end{document}